\documentclass[10pt, a4paper, reqno]{article}
\usepackage{fullpage}
\usepackage{mathtools}

\usepackage{enumitem}
\usepackage{hyperref}
\usepackage{amsthm,amsfonts,amssymb,euscript,mathrsfs}

\usepackage{color}
\usepackage{xcolor}

\newcommand{\R}{\mathbb{R}}

\newcommand{\Z}{\mathbb{Z}}
\newcommand{\N}{\mathbb{N}}

\renewcommand{\d}{\mathrm{d}}
\usepackage{stmaryrd}
\date{}
\theoremstyle{plain}

\newtheorem{lem}{Lemma}[section]
\newtheorem{thm}{Theorem}[section]
\newtheorem{prop}{Proposition}[section]
\newtheorem{coro}{Corollary}[section]
\newtheorem{mydef}{Definition}[section]
\newtheorem{remark}{Remark}[section]
\numberwithin{equation}{section}

\newcommand{\ffi}{\varphi}
\newcommand{\tffi}{\widetilde{\varphi}}

\newcommand{\e}{\varepsilon}
\newcommand{\dr}{\partial}

\newcommand{\dive}{\mathrm{div}}

\newcommand{\tr}{\mathrm{tr}}

\newcommand{\f}{\mathfrak{f}}

\newcommand{\tW}{\widetilde{W}}
\newcommand{\tsigma}{\widetilde{\sigma}}

\newcommand{\D}{\mathbf{D}}

\newcommand{\la}{\lambda}

\newcommand{\GO}[1]{O\left( #1 \right)}

\renewcommand{\l}{\left\|}
\renewcommand{\r}{\right\|}

\newcommand{\half}{\frac{1}{2}}
\newcommand{\enstq}[2]{\left\{#1~\middle|~#2\right\}}

\newcommand{\saut}{\par\leavevmode\par}

\newcommand{\h}{\mathfrak{h}}

\newcommand{\simf}{\;\overset{\mathrm{osc}}{\sim}\;}

\newcommand{\1}{\mathbf{1}}
\newcommand{\2}{\mathbf{2}}
\newcommand{\Db}{\bar{\mathbf{D}}}

\title{High-frequency solutions to the constraint equations}
\author{Arthur Touati
 \thanks{Institut des Hautes Etudes Scientifiques, Bures-sur-Yvette, France (\href{mailto:touati@ihes.fr}{touati@ihes.fr})}
}

\begin{document}

\maketitle

\begin{abstract}
We construct high-frequency initial data for the Einstein vacuum equations in dimension 3+1 by solving the constraint equations on $\R^3$. Our family of solutions $(\bar{g}_\lambda,K_\lambda)_{\lambda\in(0,1]}$ is defined through a high-frequency expansion similar to the geometric optics approach and converges in a particular sense to the data of a null dust. In order to solve the constraint equations, we use their conformal formulation and the main challenge of our proof is to adapt this method to the high-frequency context. In particular, the parameters of the conformal formulation are oscillating. The main application of this article is our companion article \cite{Touati2022a} where we construct high-frequency gravitational waves in generalised wave gauge. 
\end{abstract}

\tableofcontents

\section{Introduction}

In this article, we construct high-frequency initial data for the Einstein vacuum equations
\begin{align}
R_{\mu\nu}(g)=0,   \label{EVE chap 2}
\end{align}
where $R_{\mu\nu}(g)$ is the Ricci tensor of $g$, a Lorentzian metric on the manifold $[0,1]\times \R^3$. Initial data for \eqref{EVE chap 2} on $\Sigma_0 = \{t=0\}$ are given by a Riemannian metric $\bar{g}$ and a symmetric 2-tensor $K$. In the spacetime $([0,1]\times \R^3,g)$ that $(\Sigma_0,\bar{g},K)$ gives birth to, $\bar{g}$ will be the restriction of $g$ to $\Sigma_0$ and $K$ will be the second fundamental form of $\Sigma_0$, that is $K=-\half \mathcal{L}_T g$ where $T$ is the unit normal to $\Sigma_0$ for $g$. A necessary condition for $(\bar{g},K)$ to be the set of initial data to a solution of \eqref{EVE chap 2} is that $(\bar{g},K)$ solves the following \textit{vacuum constraint equations}:
\begin{align}
R(\Bar{g})+(\tr_{\Bar{g}}K)^2-|K|^2_{\Bar{g}} & = 0,   \label{hamiltonian constraint general}
\\ -\dive_{\Bar{g}} K +\d\tr_{\Bar{g}}K & = 0,\label{momentum constraint general}
\end{align}
where $R(\bar{g})$ denotes the scalar curvature of $\bar{g}$. Equation \eqref{hamiltonian constraint general} is the Hamiltonian constraint, and equation \eqref{momentum constraint general} is the momentum constraint.  Together, they form a coupled system of non-linear elliptic partial differential equations and we refer to Chapter 7 of \cite{ChoquetBruhat2009} for a complete presentation of their mathematical study.

\subsection{Backreaction for the constraint equations}

We construct a family $(\bar{g}_\la,K_\la)_{\la\in (0,1]}$ of solutions to \eqref{hamiltonian constraint general}-\eqref{momentum constraint general} which oscillate at frequency $\la^{-1}$. The mean feature of this family is its high-frequency limit, i.e its behaviour when $\la$ tends to 0. It converges in the following weak sense
\begin{align}
\bar{g}_\lambda& \to  \bar{g}_0,\quad\text{uniformly in $L^\infty$,}\nonumber
\\ \nabla\bar{g}_\lambda& \rightharpoonup \nabla\bar{g}_0,\quad\text{weakly in $L^2_{loc}$,} \label{behaviour}
\\ K_\lambda & \rightharpoonup K_0,\quad\text{weakly in $L^2_{loc}$,}\nonumber
\end{align} 
to a fixed solution $(\bar{g}_0,K_0)$ of the \textit{null dust maximal constraint equations}:
\begin{align}
R(\Bar{g}_0) - |K_0|^2_{\Bar{g}_0} & = 2 |\nabla u_0|^2_{\bar{g}_0}F_0^2 ,   \label{hamiltonian background}
\\ -\dive_{\Bar{g}_0} K_0  & = |\nabla u_0|_{\bar{g}_0}F_0^2\d u_0,\label{momentum background}
\\ \tr_{\Bar{g}_0}K_0 & = 0,\label{maximal}
\end{align}
where $u_0$ and $F_0$ are scalar functions defined on $\R^3$. Therefore, this article provides the first example of \textit{backreaction} for the constraint equations: a sequence of solutions to the \textit{vacuum} equations describes, in the high-frequency limit, a matter model. 

\saut
This phenomenon is at the heart of Burnett's conjecture in general relativity. In \cite{Burnett1989}, he conjectured that if a sequence $(g_\la)_\la$ of solutions to \eqref{EVE chap 2} converges weakly in $H^1_{loc}$ to a metric $g_0$, then $g_0$ solves the massless Einstein-Vlasov system
\begin{align}
R_{\alpha\beta}(g_0) & = \int_{g_0^{-1}(p,p)=0} f(x,p) p_\alpha p_\beta \d\mu_{g_0}, \label{Einstein Vlasov 1}
\\ p^\alpha \dr_\alpha f - p^\alpha p^\beta \Gamma(g_0)^\rho_{\alpha\beta} \dr_{p^\rho}f & =0, \label{Einstein Vlasov 2}
\end{align} 
for a density $f$ defined on the tangent bundle and where $\d\mu_{g_0}$ is the measure on the tangent bundle. The second equation is the Vlasov equation for the density of massless particles $f$. In his article, Burnett also asked the reverse question : given $g_0$ a solution to \eqref{Einstein Vlasov 1}-\eqref{Einstein Vlasov 2}, can one construct a sequence $(g_\la)_\la$ of solutions to \eqref{EVE chap 2} converging weakly in $H^1$ to $g_0$ ? The present article provides a positive answer to this question at the level of the spacelike data for \eqref{EVE chap 2} and for a discrete version of the massless Einstein-Vlasov system, namely the Einstein-null dust system
\begin{align}
R_{\alpha\beta}(g_0) & = F_0^2 \dr_\alpha u_0 \dr_\beta u_0, \label{ND1}
\\ g_0^{-1}(\d u_0, \d u_0) & =0, \label{ND2}
\\ 2 g_0^{\rho\sigma}\dr_\rho u_0 \dr_\sigma F_0 + (\Box_{g_0} u_0) F_0 & = 0. \label{ND3}
\end{align}
The equations \eqref{hamiltonian background}-\eqref{maximal} are the constraint equations that spacelike data for the system \eqref{ND1}-\eqref{ND3} need to solve, in the particular case where the initial hypersurface is maximal. This explains why we refered to the system \eqref{hamiltonian background}-\eqref{maximal} as the null dust maximal constraint equations. Note that we slightly abuse notations: in this article the scalar functions $F_0$ and $u_0$ will be defined on $\Sigma_0$ only, while in \eqref{ND1}-\eqref{ND3} they are defined on a whole spacetime.

\saut
The two aspects of Burnett's conjecture, the direct one and the indirect one, have been studied in several works and different settings. Assuming $\mathbb{U}(1)$ symmetry, Huneau and Luk address the direct conjecture by means of microlocal defect measure in \cite{Huneau2019} (see also \cite{Guerra2021}). Under the same symmetry, they construct high-frequency vacuum spacetimes converging to $N$ null dusts in \cite{Huneau2018a}. The first result without symmetry assumptions was obtained by Luk and Rodnianski in \cite{Luk2020}, where they address both sides of Burnett's conjecture in double null gauge. This choice of gauge restricts the class of kinetic spacetimes they consider to 2 null dusts. 

\saut
In a companion paper \cite{Touati2022a}, the author constructs high-frequency vacuum spacetimes in generalised wave gauge, paving the way to a proof of Burnett's conjecture in this gauge. This motivates the need for the high-frequency initial data solving \eqref{hamiltonian constraint general}-\eqref{momentum constraint general} provided by this article. The link between this article and \cite{Touati2022a} and its mathematical implications will be further discussed in Section \ref{section application spacetime}, after we state the main result in Section \ref{section main result}.

\saut
Finally, note that Burnett's conjecture and the backreaction it describes for the Einstein vacuum equations falls into the widest category of non-linear effects induced by homogenization, that is through the interaction of multiple scales. This type of effects can be encountered in the study of virtually all non-linear equations coming from physics. Examples are porous media, hydrodynamics, quantum mechanics etc. We refer to \cite{Tartar2009} for a rich presentation of the field of homogenization.

\subsection{The conformal method}

The constraint equations \eqref{hamiltonian constraint general}-\eqref{momentum constraint general} are underdetermined, and in order to solve them we use the conformal method. Introduced by Lichnerowicz in \cite{Lichnerowicz1944}, this method is based on a conformal formulation of the data $(\bar{g},K)$ and it identifies free parameters. In particular, this method transforms \eqref{hamiltonian constraint general}-\eqref{momentum constraint general} into a determined system of equations composed of a vectorial equation and a scalar equation. The idea giving its name to the method is to prescribe the conformal class of $\bar{g}$, i.e to fix a Riemannian metric $\gamma$ on $\R^3$ and to solve for the scalar function $\ffi$ such that
\begin{align}
\Bar{g} & =\ffi^4\gamma \label{g bar}.
\end{align}
We say that a tensor is a TT-tensor if it is traceless divergence free and symmetric and one can show that the space of TT-tensor depends only on the conformal class of the metric. Therefore, the next step in the conformal method is to decompose the symmetric 2-tensor $K$ in connection with its trace and divergence features. More precisely we fix a scalar function $\tau$ and a TT-tensor $\sigma$ for the metric $\gamma$ and solve for the vector field $W$ such that
\begin{align}
K & = \ffi^{-2}(\sigma+L_\gamma W)+\frac{1}{3}\ffi^4\gamma\tau\label{K}
\end{align}
where $L_\gamma W$ is defined in \eqref{def L W}. The exponents appearing in \eqref{g bar} and \eqref{K} are linked to the dimension of the manifold, here 3, see Chapter 6 of \cite{ChoquetBruhat2009} for their general expression. Once the parameters $\gamma$, $\tau$ and $\sigma$ are chosen, the constraint equations \eqref{hamiltonian constraint general}-\eqref{momentum constraint general} rewrite as the following coupled system of non-linear elliptic equations for $(\ffi,W)$:
\begin{align}
8\Delta_{\gamma}\ffi &= R(\gamma)\ffi+\frac{2}{3}\tau^2\ffi^5 - \left|\sigma+L_\gamma W\right|^2_{\gamma}\ffi^{-7}, \label{hamiltonian constraint}
\\ \dive_{\gamma}L_{\gamma}W & =\frac{2}{3}\ffi^6\d\tau \label{momentum constraint}.
\end{align}
A wealth of literature has been produced over the years on the construction of solutions to \eqref{hamiltonian constraint}-\eqref{momentum constraint} on various Riemannian manifolds. Note that if the mean curvature $\tau$ is constant (CMC setting), then \eqref{hamiltonian constraint} and \eqref{momentum constraint} decouple and the construction of solutions is simplified. The results on \eqref{hamiltonian constraint}-\eqref{momentum constraint} can thus be categorized into CMC, near CMC (where $\frac{\d \tau}{\tau}$ is small) or far from CMC results, where no assumption is made on $\tau$. On compact manifolds, CMC results are obtained in \cite{Isenberg1995}. Near CMC results on asymptotically Euclidean manifolds are obtained in \cite{ChoquetBruhat2000a}, where the authors also treat the constraint equations with matter. For far from CMC results, we refer the reader to \cite{Maxwell2009} and \cite{Dilts2014}, where the case of compact manifolds and asymptotically Euclidean manifolds are respectively treated.

\saut
Our main challenge is to adapt the conformal method to the high-frequency context, i.e choose the parameters so that the resulting solution $(\bar{g}_\la,K_\la)$ displays the behaviour \eqref{behaviour}. Not only will the solutions $(\ffi,W)$ be defined by high-frequency ansatz, but also the parameters $\gamma$, $\tau$ and $\sigma$. In particular, we will construct an oscillating TT-tensor, i.e solve $\dive_\gamma \sigma = 0$ and $\tr_\gamma \sigma=0$, with $\gamma$ itself oscillating.

\subsection{Preliminaries}

We fix here our notations and present the analytic setting of this article. 

\subsubsection{Geometric notations}

Throughout this article, the notation $\Sigma_0$ refers to the manifold $\R^3$. On $\Sigma_0$ we consider the usual Euclidean coordinates $x=(x^1,x^2,x^3)$, and we denote by $e$ the standard Euclidean metric, i.e
\begin{align*}
e= \left(\d x^1\right)^2 + \left(\d x^2\right)^2 + \left(\d x^3\right)^2.
\end{align*}
Latin indices are used for the Euclidean coordinates and therefore runs from 1 to 3. In this article, repeated indices are always summed over. A second frame adapted to the background structure will be defined in Section \ref{section bg metric}.

\saut
 If $f$ is a scalar function, we define its gradient by $\nabla f=(\dr_1 f ,  \dr_2 f , \dr_3 f)$.  If $h$ is a Riemannian metric on $\Sigma_0$ we define 
\begin{align*}
|\nabla f|_h^2 = h^{ij}\dr_i f \dr_j f.
\end{align*}
In the particular case of the Euclidean metric we simply write $|\nabla f|^2$.
Moreover if $T$ and $S$ are two symmetric 2-tensors on $\Sigma_0$, we define $| T\cdot S|_h=h^{ij}h^{k\ell}T_{ik}S_{j\ell}$ and $|T|^2_h=|T\cdot T|_h$. The trace of $T$ with respect to $h$ is defined by $\tr_h T=h^{ij}T_{ij}$.

\subsubsection{Function spaces and asymptotically Euclidean manifold}\label{section function spaces chapter 2}

If $x\in \R^3$ we set $\langle x  \rangle\vcentcolon = (1+|x|^2)^\half$ with $|x|=\sqrt{(x^1)^2 + (x^2)^2 + (x^3)^2}$. We define the following weighted Sobolev spaces on $\R^3$.

\begin{mydef}[Weighted Sobolev spaces]
For $1\leq p < + \infty$, $\delta\in\R$ and $k\in\N$ we define the space $W^{k,p}_\delta$ as the completion of $C^\infty_c$ for the norm
\begin{align*}
\l u\r_{W^{k,p}_\delta} & = \sum_{0\leq |\alpha|\leq k}\l  \langle x  \rangle^{\delta+|\alpha|} \nabla^\alpha u   \r_{L^p},
\end{align*}
where the $L^p$ norm is defined with the Euclidean volume element. We extend this definition to tensors of any type by summing over all components in the Euclidean coordinates.  Some special cases are $H^k_\delta \vcentcolon = W^{k,2}_\delta$ and $L^p_\delta \vcentcolon = W^{0,p}_\delta$.
\end{mydef}

We also define the following $L^\infty$-based spaces:
\begin{mydef}
For $k\in\N$ and $\delta\in\R$ we define $C^k_\delta$ as the completion of $C^\infty_c$ for the norm
\begin{align*}
\l u\r_{C^{k}_\delta} & = \sum_{0\leq |\alpha|\leq k}\l  \langle x  \rangle^{\delta+|\alpha|} \nabla^\alpha u   \r_{L^\infty},
\end{align*}
\end{mydef}
In the following proposition we recall some useful facts about these spaces (see \cite{ChoquetBruhat2009} for the proofs).
\begin{prop}\label{prop WSS chap 2}
Let $s,s',s_1,s_2,m\in\N$, $\delta,\delta',\delta_1,\delta_2,\beta\in\R$ and $1\leq p<+\infty$.
\begin{enumerate}
\item If $s\leq \min(s_1,s_2)$, $s<s_1+s_2 - \frac{3}{p}$ and $\delta<\delta_1+\delta_2+\frac{3}{p}$ we have the continuous embedding
\begin{align*}
W^{s_1,p}_{\delta_1}\times W^{s_2,p}_{\delta_2} \subset W^{s,p}_\delta.
\end{align*}
\item If $m<s-\frac{3}{p}$ and $\beta\leq\delta + \frac{3}{p}$ we have the continuous embedding
\begin{align*}
W^{s,p}_\delta \subset C^m_\beta.
\end{align*}
\end{enumerate}
\end{prop}

The background metric and the solution of the constraint equations we will produce are asymptotically Euclidean, meaning that they converge in some sense to the Euclidean metric at infinity. We give the definition of \cite{ChoquetBruhat2000a}.

\begin{mydef}[Asymptotically Euclidean initial data]\label{def asymp}
Let $\bar{g}$ be a metric on $\R^3$ and $K$ a symmetric 2-tensor on $\R^3$.  If $k>\frac{5}{2}$ and $\delta>-\frac{3}{2}$, we say that the pair $(\bar{g} , K)$ is $H^k_\delta$ asymptotically flat if
\begin{align*}
\bar{g}-e \in H^k_\delta \quad \text{and} \quad K \in H^{k-1}_{\delta+1}.
\end{align*}
\end{mydef}
Note that the restrictions on $k$ and $\delta$ in the previous definition ensure that $\bar{g}$ is $C^1$ and $\bar{g}-e$ tends to 0 at infinity, since Proposition \ref{prop WSS chap 2} implies $H^k_\delta \xhookrightarrow{} C^1_{\delta+\frac{3}{2}}$.

\begin{remark}\label{remark ADM 1}
While the assumptions in Definition \ref{def asymp} are adapted to the inversion of elliptic operators, the decay assumption $\delta>-\frac{3}{2}$ is too weak to be able to define the ADM mass of $\bar{g}$. Introduced in \cite{Arnowitt2008}, the ADM mass of an asymptotically Euclidean manifold is defined by
\begin{align*}
M(\bar{g})= \frac{1}{16\pi} \lim_{R\to +\infty} \int_{\dr B_R} \sum_{i,j}\left( \dr_i \bar{g}_{ij} - \dr_j \bar{g}_{ii}  \right) \nu_j \d\mu ,
\end{align*}
where $\nu$ is the outgoing normal to the spheres $\dr B_R$ and $\d\mu$ the measure on these spheres. Since the volume of the spheres $\dr B_R$ grows like $R^2$, we need to ask for $\delta\geq-\half$ in Definition \ref{def asymp} in order to be able to define $M(\bar{g})$ in general. See also Remark \ref{remark ADM 2}.
\end{remark}

\subsubsection{Elliptic estimates}\label{section elliptic estimates}

Solving the constraint equations with the conformal method requires to invert the elliptic operators $\Delta_\gamma$ and $\dive_\gamma L_\gamma$, for $\gamma$ a Riemannian metric. The first one is the Laplace-Beltrami operator acting on scalar functions and the second one is the conformal Laplacian acting on vector field, see Section \ref{section exp diff op} for their exact definition. The metric $\gamma$ will be defined in Section \ref{section def gamma} but we can already say that it will be close to a background metric $\bar{g}_0$, itself close to the Euclidean metric $e$ on $\R^3$. We will benefit from this fact and invert $\Delta=\Delta_e$ and $\dive_e L_e$ rather than $\Delta_\gamma$ and $\dive_\gamma L_\gamma$, which have oscillating coefficients. The following proposition gives the desired inversion properties. The proof of its first part can be found in \cite{McOwen1979} while the second part is proved in \cite{ChoquetBruhat2000a}.

\begin{prop}\label{prop Delta et dive e L e}
If $-\frac{3}{2}<\delta< -\half$ then $\Delta : H^2_\delta \longrightarrow L^2_{\delta+2}$ and $\dive_eL_e : H^2_\delta \longrightarrow L^2_{\delta+2}$ are isomorphisms.
\end{prop}

\subsubsection{High-frequency notations}\label{section HF notations}

In this article we consider high-frequency quantities, i.e tensors of all types including scalar functions, metrics, 1-forms etc.  A quantity is said to be high-frequency if it admits an expansion in powers of the small parameter $\lambda$ with coefficients of the form
\begin{align}
\mathrm{T} \left( \frac{u_0}{\lambda}\right) f \label{T u0 lambda f}
\end{align}
where $f$ depends only on $x\in\Sigma_0$ and $\mathrm{T}$ is an oscillating function, i.e an element of
\begin{align}
\enstq{\theta\in \R\longmapsto\sin(k\theta)}{k\in\N} \cup \enstq{\theta\in \R\longmapsto\cos(k\theta)}{k\in\N} .\label{ensemble fcts trigo}
\end{align}
When considering a high-frequency quantity such as a tensor $S$, we denote by $S^{(i)}$ the coefficients of $\lambda^i$ in the expansion defining $S$, which thus expands formally as
\begin{align*}
S = \sum_{i\in\Z}\lambda^i S^{(i)}.
\end{align*}
Note that $S^{(i)}$ is a tensor of the same type as $S$. Moreover,  if $j\in \Z$ we define $S^{(\geq j)}$ by $S^{(\geq j)} = \sum_{k\geq j} \lambda^{k-j}S^{(k)}$.  This allows us to clearly troncate high-frequency expansions at a fixed order as in
\begin{align*}
S= \sum_{k\leq j-1}\lambda^k S^{(k)} + \lambda^j S^{(\geq j)}.
\end{align*}
To emphasize the fact that a high-frequency coefficient $S^{(i)}$ of a tensor $S$ is oscillating we will often write $S^{(i)}\left(\frac{u_0}{\lambda}\right)$ instead of just $S^{(i)}$.

\saut
In order to describe concisely the \textit{oscillating behaviour} of a high-frequency coefficient $S^{(i)}$ of a tensor $S$, we write for $\mathcal{A}$ a finite subset of \eqref{ensemble fcts trigo}
\begin{align*}
S^{(i)} \simf \sum_{\mathrm{T}\in \mathcal{A}} \mathrm{T}(\theta)
\end{align*}
if there exists tensors $(S^{(i)}_{\mathrm{T}})_{\mathrm{T}\in \mathcal{A}}$ of the same type as $S^{(i)}$ such that
\begin{align*}
S^{(i)} = \sum_{\mathrm{T}\in \mathcal{A}} \mathrm{T} \left( \frac{u_0}{\lambda}\right) S^{(i)}_{\mathrm{T}}.
\end{align*}
This notation allows us to compute the oscillating behaviour of non-linear quantities without caring too much about the non-oscillating coefficients $S^{(i)}_{\mathrm{T}}$. Note that $S^{(i)}\simf 1$ simply means that $S^{(i)}$ is non-oscillating, i.e does not depend on $\frac{u_0}{\lambda}$.  

\saut
In terms of derivation, we use the symbol $\theta$ for the derivation with respect to the $\frac{u_0}{\lambda}$ variable. For example, if $f$ is a scalar function and if $g=\mathrm{T} \left( \frac{u_0}{\lambda}\right) f$ then 
\begin{align*}
\dr_\theta g & = \mathrm{T}' \left( \frac{u_0}{\lambda}\right) f.
\end{align*}

\subsection{Acknowledgment}

The author would like to thank Caterina Vâlcu for all the helpful discussions.

\section{Statement of the results}\label{section statement result chap 2}

In this section, we give the assumptions on the background, state our main result and discuss its main application.

\subsection{The background}\label{section bg metric}

The background metric and 2-tensor $(\bar{g}_0,K_0)$ solve \eqref{hamiltonian background}-\eqref{maximal} with $F_0$ and $u_0$ two scalar functions defined on $\R^3$.  The full background solution is then $(\bar{g}_0, K_0,F_0,u_0)$ and we make some assumptions on it.
\begin{itemize}
\item \textbf{Assumptions on $(\bar{g}_0,K_0)$.} Even though $(\bar{g}_0,K_0)$ satisfies non-vacuum constraint equations, the sources are compactly supported (see \eqref{support F0} below) and we assume that $(\bar{g}_0,K_0)$ corresponds to asymptotically Euclidean and highly regular initial data. By this we mean that there exists $\delta>-\frac{3}{2}$, a large integer $N\geq 10$ and $\e>0$ such that
\begin{align}
\l \bar{g}_0 - e \r_{H^{N+1}_\delta} + \l K_0 \r_{H^{N}_{\delta+1}} \leq \e.\label{estim g0 K0}
\end{align}
We denote by $\Db$ the covariant derivative associated to $\bar{g}_0$.
\item \textbf{Assumptions on $F_0$.} The density $F_0$ is supported in a ball of size $R>0$ in $\R^3$, i.e
\begin{align}
\mathrm{supp}(F_0)\subset B_R\vcentcolon =\enstq{x\in\R^3}{  |x|\leq R  }.\label{support F0}
\end{align}
It also enjoys some regularity:
\begin{align}
\l F_0 \r_{H^{N}} \leq \e\label{estim F0 chap 2}
\end{align}
where $\e$ is defined above. 
\item \textbf{Assumptions on $u_0$.} There exists a constant non-zero vector field $\mathfrak{z}=(\mathfrak{z}_1,\mathfrak{z}_2,\mathfrak{z}_3)$ such that 
\begin{align}
\l \nabla u_0 - \mathfrak{z}  \r_{H^{N}_{\delta+1}}\leq \e\label{estim u0}
\end{align}
where $\nabla u_0 = (\dr_1 u_0, \dr_2 u_0, \dr_3 u_0)$ is the euclidean gradient of $u_0$. By taking $\e$ small enough in \eqref{estim u0} we can assume that $|\nabla u_0|$ is uniformly bounded from below, which implies that $u_0$ has no critical points. Moreover, the level hypersurfaces of $u_0$ defined by
\begin{align*}
P_{0,u}=\enstq{x\in \R^3}{u_0(x)=u}
\end{align*}
foliates $\R^3$ and have the topology of planes thanks to \eqref{estim u0}. This allows us to define a particular frame at each point of $\R^3$. We define the vector field
\begin{align*}
N_0 = - \frac{\bar{g}_0^{ij}\dr_i u_0 \dr_j}{| \nabla u_0 |_{\bar{g}_0}}.
\end{align*}
It satisfies $\bar{g}_0(N_0,N_0)=1$ and is orthogonal to the hypersurface $P_{0,u}$. We consider at each point $x\in\R^3$ an orthonormal basis $(e_\1,e_\2)$ of $T_xP_{0,u}$ for the metric $\bar{g}_0$. The frame $(N_0,e_\1,e_\2)$ will play a crucial role in our construction. While we reserve the usual latin indices for the coordinates system $(x^1,x^2,x^3)$, i.e $i,j\in\{1,2,3\}$, the bold latin indices are used for the frame $(e_\1,e_\2)$, i.e $\mathbf{i}, \mathbf{j} \in \{ \1,\2\}$.
\end{itemize}
In this article, we don't prove that a background solution $(\bar{g}_0,K_0,F_0,u_0)$ solving \eqref{hamiltonian background}-\eqref{maximal} and satisfying the above assumptions exists. We refer to \cite{ChoquetBruhat2000a} for the details of how one can solve the constraint equations with sources in the asymptotically Euclidean setting.

\subsection{Solving the constraint equations}\label{section main result}

The following theorem is the main result of this article.

\begin{thm}\label{theo main chap 2}
Let $(\bar{g}_0,K_0,F_0,u_0)$ be the solution of the maximal constraint equations coupled with a null dust described in Section \ref{section bg metric}, and let $\e>0$ be the smallness threshold. There exists $\e_0=\e_0(\delta,R)>0$ such that if $0<\e\leq \e_0$, there exists for all $\lambda\in (0,1]$ a solution $(\bar{g}_\lambda,K_\lambda)$ solution of the vacuum constraint equations \eqref{hamiltonian constraint general}-\eqref{momentum constraint general} on $\R^3$ of the form
\begin{align}
\bar{g}_\lambda & = \bar{g}_0 + \lambda \cos\left( \frac{u_0}{\lambda}\right) \bar{F}^{(1)} + \lambda^2 \left( \sin\left( \frac{u_0}{\lambda}\right) \bar{F}^{(2,1)} + \cos\left( \frac{2u_0}{\lambda}\right) \bar{F}^{(2,2)} \right) + \lambda^2 \bar{\h}_\lambda,\label{g bar theo chap 2}
\\ K_\lambda & = K^{(0)}_\lambda + \lambda K^{(1)}_\lambda + \lambda^2 K^{(\geq 2)}_\lambda, \label{K lambda theo chap 2}
\end{align}
with
\begin{align}
K^{(0)}_\lambda & = K_0 + \half \sin\left( \frac{u_0}{\lambda}\right) |\nabla u_0|_{\bar{g}_0} \bar{F}^{(1)},\label{Klambda 0}
\\ K^{(1)}_\lambda & = \cos\left(\frac{u_0}{\lambda}\right)K^{(1,1)} + \sin\left(\frac{2u_0}{\lambda}\right)K^{(1,2)}.  \label{Klambda 1}
\end{align}
Moreover:
\begin{itemize}
\item[(i)] the tensors $\bar{F}^{(1)}$, $\bar{F}^{(2,1)}$ and $\bar{F}^{(2,2)}$ are supported in $\{|x|\leq R\}$ and there exists $C_{\mathrm{cons}}=C_{\mathrm{cons}}(\delta,R)>0$ such that
\begin{align*}
\l \bar{F}^{(1)} \r_{H^N} + \l \bar{F}^{(2,1)}  \r_{H^{N-1}} + \l \bar{F}^{(2,2)} \r_{H^{N-1}} & \leq C_{\mathrm{cons}} \e,
\end{align*}
\item[(ii)] the tensor $\bar{F}^{(1)}$ is $\bar{g}_0$-traceless, tangential to $P_{0,u}$ and satisfies
\begin{align}
\left| \bar{F}^{(1)} \right|^2_{\bar{g}_0} = 8 F_0^2 \label{energie condition theo chap 2},
\end{align}
\item[(iii)] $K^{(1,1)}$ and $K^{(1,2)}$ are given by
\begin{align}
K^{(1,1)} & =   \half N_0\bar{F}^{(1)}_{ij} + \half \left( \bar{g}_0^{k\ell}(K_0)_{(i\ell} - N_0^\ell \Gamma(\bar{g}_0)^k_{\ell (i} \right) \bar{F}^{(1)}_{j)k} \label{K 11 theo}
\\&\quad  + \frac{1}{2|\nabla u_0|_{\bar{g}_0}} \left(  K_0^{k\ell}   \dr_k u_0 \dr_\ell u_0  +   \bar{g}_0^{k\ell} \dr_k u_0  \dr_\ell |\nabla u_0|_{\bar{g}_0}  - \half \bar{g}_0^{k\ell}\dr_k\dr_\ell u_0 \right) \bar{F}^{(1)}_{ij}  \nonumber
\\&\quad  -\half   \left|\nabla u_0 \right|_{\bar{g}_0} \bar{F}^{(2,1)},\nonumber
\\ K^{(1,2)} & =  \left|\nabla u_0 \right|_{\bar{g}_0} \bar{F}^{(2,2)},\label{K 12 theo}
\end{align}
\item[(iv)] the tensors $\bar{\h}_\lambda$ and $K^{(\geq 2)}_\lambda$ belong to the spaces $H^5_\delta$ and $H^4_{\delta+1}$ respectively and satisfy
\begin{align}
\max_{ r\in\llbracket 0,4\rrbracket} \lambda^r \l \nabla^{r+1} \bar{\h}_\lambda \r_{L^2_{\delta+r+1}}&  \leq C_{\mathrm{cons}} \e ,\label{estim h bar}
\\ \max_{r\in\llbracket 0,4\rrbracket } \lambda^r \l \nabla^r K_\lambda^{(\geq 2)} \r_{L^2_{\delta+r+1}}&  \leq C_{\mathrm{cons}} \e . \label{estim K geq 2}
\end{align}
\end{itemize}
\end{thm}

Note that in Theorem \ref{theo main chap 2}, the only free data are the background quantities $(\bar{g}_0,K_0,F_0,u_0)$. The tensors $\bar{F}^{(1)}$, $\bar{F}^{(2,1)}$ and $\bar{F}^{(2,2)}$ are obtained in the proof and are determined by $(\bar{g}_0,K_0,F_0,u_0)$. The only freedom lies in the choice of the polarization of $\bar{F}^{(1)}$, see the definition of $\omega^{(1)}$ in Section \ref{section def gamma}.

\begin{remark}\label{remark ADM 2}
Note that if the background metric $\bar{g}_0$ satisfies a stronger decay assumption, one can define its ADM mass (see Remark \ref{remark ADM 1}). In this case, it is also possible to show the convergence of the ADM mass of $\bar{g}_\la$ towards the ADM mass of the background when $\la$ tends to 0. Indeed, the high-frequency perturbations in \eqref{g bar theo chap 2} are compactly supported and only have an impact at infinity through the remainder $\bar{\h}_\lambda$, for which we can show enough control in terms of decay and when $\la$ tends to 0.
\end{remark}

\subsection{Application to high-frequency vacuum spacetimes}\label{section application spacetime}

Even though Theorem \ref{theo main chap 2} is formulated in a purely elliptic way, its main purpose is to provide spacelike initial data for the high-frequency solutions to \eqref{EVE chap 2} constructed in \cite{Touati2022a} in generalised wave gauge. More precisely, in this article we construct a family of Lorentzian metric $(g_\la)_{\la\in (0,1]}$ of the form
\begin{align}
g_\lambda = g_0 + \lambda g^{(1)} \left(\frac{u_0}{\lambda}\right) + \lambda^2 g^{(2)} \left(\frac{u_0}{\lambda}\right) + \GO{\lambda^2}   \label{ansatz g rough}
\end{align}
where the $g^{(i)}$ are periodic and smooth functions of their argument $\frac{u_0}{\lambda}$ and where $(g_0,F_0,u_0)$ is a given solution of the Einstein-null dust system \eqref{ND1}-\eqref{ND3} with initial data on $\Sigma_0$ given by $(\bar{g}_0, K_0, F_0, u_0)$ described in Section \ref{section bg metric}. In this section, let us discuss how it affects the construction of Theorem \ref{theo main chap 2}.

\saut
Since $\bar{g}_\la$ is the induced metric of $g_\la$ on $\Sigma_0$, the oscillating behaviours of $g^{(1)}$ and $g^{(2)}$ in \cite{Touati2022a} explains the high-frequency expansion of $\bar{g}_\la$ in \eqref{g bar theo chap 2}. In particular, we have $g^{(1)}=\cos\left(\frac{u_0}{\la}\right)F^{(1)}$ with $F^{(1)}\restriction{\Sigma_0}=\bar{F}^{(1)}$. The fact that $\bar{F}^{(1)}$ is $\bar{g}_0$-traceless and tangential to $P_{0,u}$ translates as an initial polarization condition for $F^{(1)}$ equivalent to $R_{\mu\nu}(g_\la)=\GO{\la^0}$.

\begin{remark}\label{remark TT gauge}
These conditions correspond exactly to the definition of the TT gauge in the linearized gravity setting for a plane wave propagating in the $N_0$ direction. Note that $N_0$ is not a constant vector field, therefore strictly speaking a plane wave can't propagate in the $N_0$ direction. However, as \eqref{estim g0 K0} and \eqref{estim u0} show, $N_0$ is close to the $\mathfrak{z}$ direction and the analogy with the TT gauge of linearized gravity is thus valid.
\end{remark}

In \cite{Touati2022a}, these polarization conditions are propagated by the transport equation that $F^{(1)}$ must satisfy in the spacetime so that $R_{\mu\nu}(g_\la)=\GO{\la^1}$, namely
\begin{align}
-2\D_{L_0} F^{(1)} + (\Box_{g_0}u_0) F^{(1)} = 0,\label{transport F1 spacetime}
\end{align}
where $\D$ is the covariant derivative associated to $g_0$ and $L_0=-g_0^{\mu\nu}\dr_\mu u_0 \dr_\nu$ is the spacetime gradient of $u_0$. Moreover, in \cite{Touati2022a} we assume without loss of generality that $\dr_t$ is the unit normal to $\Sigma_0$ for $g_0$ and that $u_0$ solves the eikonal equation (see \eqref{ND2}), which thus give $\dr_t u_0 = |\nabla u_0|_{\bar{g}_0}$ and $L_0 = |\nabla u_0|_{\bar{g}_0}(\dr_t + N_0)$ on $\Sigma_0$. This also implies that on $\Sigma_0$ we have
\begin{align}
\dr_t \left( \la g^{(1)}_{ij} \right) = -\sin\left(\frac{u_0}{\lambda}\right) |\nabla u_0|_{\bar{g}_0} F^{(1)}_{ij}  + \lambda \cos\left(\frac{u_0}{\lambda}\right) \dr_t F^{(1)}_{ij} \label{dt g1 initialement}
\end{align}
where $\dr_t F^{(1)}_{ij}$ on $\Sigma_0$ is directly given by the first order transport equation \eqref{transport F1 spacetime}. Note in particular that thanks to our previous geometric assumption on $\dr_t$ and $u_0$ on $\Sigma_0$ and thanks to the wave gauge condition that $g_0$ is also assumed to satisfy in the spacetime, we can completely rewrite $\dr_t F^{(1)}_{ij}\restriction{\Sigma_0}$ given by \eqref{transport F1 spacetime} in terms of $\bar{g}_0$, $K_0$ and $\bar{F}^{(1)}$. For instance, we can show that
\begin{align*}
-\half\Box_{g_0}u_0\restriction{\Sigma_0} = (K_0)^{k\ell}   \dr_k u_0 \dr_\ell u_0  +   \bar{g}_0^{k\ell} \dr_k u_0  \dr_\ell |\nabla u_0|_{\bar{g}_0}  - \half \bar{g}_0^{k\ell}\dr_k\dr_\ell u_0.
\end{align*}
A similar treatment can be applied to all the terms in \eqref{transport F1 spacetime}. Now, on the one hand, $K_\la$ must be the second fundamental form $-\half \mathcal{L}_{T_\la} g_\la$ of $\Sigma_0$, where $T_\la$ is the unit normal to $\Sigma_0$ for $g_\la$. On the other hand, we choose data in \cite{Touati2022a} such that $T_\la = \dr_t + \GO{\la^2}$. Therefore, the oscillating parts of $K^{(0)}_\la$ and $K^{(1)}_\la$ must be consistent with \eqref{dt g1 initialement}, which explains expressions \eqref{Klambda 0} and \eqref{K 11 theo}.  

\begin{remark}\label{remark redundancy}
This entanglement between the construction of the data and the evolution equations in the spacetime is characteristic of the present high-frequency setting. Indeed, the spacelike constraint equations give two initial data, roughly $g\restriction{\Sigma_0}$ and $\dr_t g \restriction{\Sigma_0}$, while the oscillating terms in the spacetime ansatz satisfy first order transport equations which only require $g\restriction{\Sigma_0}$ and prescribe, in part, $\dr_t g \restriction{\Sigma_0}$.
\end{remark}

Finally, note that the use of Sobolev spaces instead of Hölder spaces in Theorem \ref{theo main chap 2} also comes from the application to \cite{Touati2022a}, since Sobolev spaces are more suited to non-linear wave equations. Similarly, the constraint $N\geq 10$ allows us to perform a bootstrap argument for the remainder in \cite{Touati2022a} using only $L^2-L^\infty$ estimates.

\section{Strategy of proof}

In this section, we present the strategy of the proof of Theorem \ref{theo main chap 2}, which consists mainly in an adaptation of the conformal method to the high-frequency setting.

\begin{remark}
From now on and until Section \ref{section conclusion chap 2} we drop the index $\lambda$ in the solution $(\bar{g}_\lambda,K_\lambda)$ obtained in Theorem \ref{theo main chap 2} and simply write $(\bar{g},K)$.
\end{remark}

\subsection{High-frequency expansion of the parameters}\label{section HF parameters unknowns}

We start by giving and motivating the expansions of the parameters of the conformal method, i.e the metric $\gamma$, the scalar function $\tau$ and the TT-tensor $\sigma$.

\saut
Since the metric $\bar{g}$ obtained in Theorem \ref{theo main chap 2} will ultimately inherit its properties from the metric $\gamma$ defining the conformal class, $\gamma$ needs needs to fulfill two requirements: it needs to oscillate and to be close to $\bar{g}_0$. Therefore, we define $\gamma$ by
\begin{align}
\gamma & = \bar{g}_0 + \lambda \gamma^{(1)}\left( \frac{u_0}{\lambda}\right) + \lambda^2\gamma^{(2)}\left( \frac{u_0}{\lambda}\right),\label{ansatz gamma}
\end{align}
where we recall that the notation $\gamma^{(i)}\left( \frac{u_0}{\lambda}\right)$ is used to emphasize the fact that $\gamma^{(i)}$ is a linear combination of terms of the form \eqref{T u0 lambda f}. See Section \ref{section def gamma} for the exact definition of $\gamma^{(1)}$ and $\gamma^{(2)}$.

\saut
Similarly, we want $K$ to be close to $K_0$ (in a weak sense since it is at the level of one derivative of the metric, see \eqref{behaviour}) and $K_0$ is assumed to be $\bar{g}_0$-traceless. We thus define $\tau$ to be of order $\lambda^1$, i.e
\begin{align}
 \tau & =  \lambda\tau^{(1)}\left( \frac{u_0}{\lambda}\right) .\label{ansatz tau}
\end{align} 
The fact that the high-frequency perturbations don't contribute to the mean curvature at the $\lambda^0$ order will be justified in the proof, but this can already be seen in Theorem \ref{theo main chap 2}: the contribution at $\la^0$ order to $K$ is given by $\half \sin\left(\frac{u_0}{\la}\right)|\nabla u_0|_{\bar{g}_0} \bar{F}^{(1)}$ (see \eqref{Klambda 0}), and $\bar{F}^{(1)}$ is $\bar{g}_0$-traceless.

\saut
The definition of the TT-tensor $\sigma$ requires a special construction linked to expressions \eqref{K 11 theo}-\eqref{K 12 theo}. Since it involves also the expansion of the solution $\ffi$ given in the next section, we postpone the discussion of this special construction to Section \ref{section strategy TT tensor}. However, we can already give its high-frequency expansion: 
\begin{align}
\sigma & = \sigma^{(0)}\left( \frac{u_0}{\lambda}\right) +\lambda \sigma^{(1)}\left( \frac{u_0}{\lambda}\right) + \lambda^2 \left( \sigma^{(2)}\left( \frac{u_0}{\lambda}\right) + \tsigma \right).\label{ansatz sigma strategy}
\end{align}
Here, $\tsigma$ is a non-oscillating remainder whose role is to solve the equations defining the class of TT-tensors, that is $\dive_\gamma \sigma =0$ and $\tr_\gamma \sigma =0$, with $\gamma$ defined by \eqref{ansatz gamma}.

\subsection{High-frequency expansion of the solutions}

The solutions $(\ffi,W)$ need to solve equations \eqref{hamiltonian constraint}-\eqref{momentum constraint}, where the parameters defined above appear as coefficients. To understand the expansions for $(\ffi,W)$, let us look at the effect of the expansions \eqref{ansatz gamma}-\eqref{ansatz tau}-\eqref{ansatz sigma strategy} on the equations \eqref{hamiltonian constraint}-\eqref{momentum constraint}. 

\saut
A standard feature of high-frequency quantities is that they lose one power of $\lambda$ per derivative. In \eqref{hamiltonian constraint}, derivatives of the metric $\gamma$ only appear in its scalar curvature $R(\gamma)$ and we can derive from \eqref{ansatz gamma} that
\begin{align}
R(\gamma) & = \frac{ |\nabla u_0|^2_{\bar{g}_0} }{\lambda} \dr^2_\theta  \left(      \gamma^{(1)}_{N_0N_0}   -  \tr_{\bar{g}_0}  \gamma^{(1)}_{ij}      \right) + \GO{\lambda^0}.\label{R(gamma) intro}
\end{align}
Therefore, if $\gamma^{(1)}$ is well-chosen, the RHS of \eqref{hamiltonian constraint} is $\GO{\lambda^0}$. This implies that the first oscillating term in $\ffi$ must be at order $\lambda^2$, since $\Delta_\gamma$ is a second order elliptic operator. More precisely we choose
\begin{align}
 \ffi & = 1 + \lambda^2 \left( \ffi^{(2)}\left( \frac{u_0}{\lambda}\right) + \tffi \right) + \lambda^3 \ffi^{(3)}\left( \frac{u_0}{\lambda}\right) ,  \label{ansatz ffi}
\end{align}
where $\tffi$ is a non-oscillating remainder. The constant coefficient in \eqref{ansatz ffi} ensures that $\bar{g}_\lambda  =\bar{g}_0 + \GO{\lambda}$. Even though $\ffi^{(3)}$ appears after the remainder $\tffi$ in the ansatz for $\ffi$, it appears before $\tffi$ in the hierarchy obtained from \eqref{hamiltonian constraint}, again since $\Delta_\gamma$ is a second order elliptic operator. 

\saut
Similarly, \eqref{ansatz tau} and \eqref{ansatz ffi} imply that $\ffi^6\d\tau=\GO{\lambda^0}$ and the ellipticity of $\dive_\gamma L_\gamma$ allows us to choose $W$ of the following form
\begin{align}
W & =  \lambda^2\left( W^{(2)}\left( \frac{u_0}{\lambda}\right) + \tW \right)+ \lambda^3 W^{(3)}\left( \frac{u_0}{\lambda}\right) ,  \label{ansatz W}
\end{align}
where $\tW$ is a non-oscillating remainder. 

\saut
As is standard when considering high-frequency expansions, we obtain a hierarchy of equations. Very schematically, one can say that $\ffi^{(2)}$ and $\ffi^{(3)}$ (resp. $W^{(2)}$ and $W^{(3)}$) will solve the orders $\lambda^0$ and $\lambda^1$ of \eqref{hamiltonian constraint} (resp. \eqref{momentum constraint}), while the non-oscillating remainder $\tffi$ (resp. $\tW$) will solve the orders higher than $\lambda^2$. 

\saut
Even though the operators $\Delta_\gamma$ and $\dive_\gamma L_\gamma$ are both second order elliptic operators, they depend differently on $\gamma$. Indeed, the coefficients of $\Delta_\gamma$ depend on $\gamma$ and $\nabla\gamma$ only, while the coefficients of $\dive_\gamma L_\gamma$ also depend on $\nabla^2\gamma$. From \eqref{ansatz gamma}, we thus see that $\dive_\gamma L_\gamma$ loses inherently one power of $\lambda$. The analysis of the last equations in the hierarchy will therefore differ from the hamiltonian to the momentum constraint.

\subsection{The TT-tensor}\label{section strategy TT tensor}

In this section, we give more details on the definition of $\sigma$, and in particular the first two terms in its high-frequency expansion \eqref{ansatz sigma strategy}, i.e $\sigma^{(0)}$ and $\sigma^{(1)}$. Thanks to \eqref{ansatz ffi} we have $\ffi= 1 + \GO{\la^2}$ which implies that
\begin{align}
K^{(0)} & = \sigma^{(0)} + (L_\gamma W)^{(0)}, \label{strategy def sigma0}
\\ K^{(1)} & = \sigma^{(1)} + (L_\gamma W)^{(1)} + \frac{1}{3}\bar{g}_0\tau^{(1)},\label{strategy def sigma1}
\end{align}
where we used \eqref{K}, \eqref{ansatz gamma} and \eqref{ansatz tau}. As explained in Section \ref{section application spacetime}, the expressions of $K^{(0)}$ and $K^{(1)}$ are actually prescribed the application of Theorem \ref{theo main chap 2} to the definition of initial data for \eqref{EVE chap 2}. Therefore, \eqref{Klambda 0}-\eqref{K 12 theo} are used to define $K^{(0)}$ and $K^{(1)}$, while $\sigma^{(0)}$ and $\sigma^{(1)}$ are defined \textit{a posteriori} such that \eqref{strategy def sigma0}-\eqref{strategy def sigma1} hold. We are then left with the task of showing that $\sigma^{(0)}+\lambda\sigma^{(1)}$ defines an approximate TT-tensor and of constructing an exact one, with the help of $\sigma^{(2)}$ and the remainder $\tsigma$ introduced in \eqref{ansatz sigma strategy}. The full construction of $\sigma$ is the content of Sections \ref{section almost TT} and \ref{section TT}.

\saut
This procedure goes against the standard conformal method, where one starts by defining the parameters, then solve for the $(\ffi,W)$, and finally obtain $(\bar{g},K)$ via \eqref{g bar}-\eqref{K}. This is another characteristic feature of the high-frequency setting, which originates in the redundancy described in Remark \ref{remark redundancy}.

\subsection{Outline of the proof}

We give here the outline of the rest of this article,  which proves Theorem \ref{theo main chap 2}.
\begin{itemize}
\item In Section \ref{section HF conformal class} we define the metric $\gamma$ and compute useful high-frequency expansions such as its scalar curvature $R(\gamma)$ or differential operators depending on $\gamma$ and appearing in the constraint equations.
\item In Section \ref{section approximate solution} we define $\ffi^{(2)}$, $\ffi^{(3)}$ and $W^{(2)}$ such that they solve the first orders of the constraint equations and by doing so we also define the parameter $\tau$. This section is concluded by the construction of the TT-tensor $\sigma$.
\item In Section \ref{section exact solution} we fully solve the constraint equations with a fixed point argument for the remainders $\tffi$ and $\tW$ (we also define $W^{(3)}$ in the process).
\item In Section \ref{section conclusion chap 2} we conclude the proof of Theorem \ref{theo main chap 2} by proving \eqref{estim h bar} and \eqref{estim K geq 2}.
\end{itemize}

\section{High-frequency conformal class}\label{section HF conformal class}

In this section we define the metric $\gamma$ on $\Sigma_0$,  that is the prefered member of the conformal class in which we look for $\bar{g}$ according to \eqref{g bar}. 

\subsection{Definitions and first computations}\label{section def gamma}

We choose 
\begin{align}
\gamma  = \bar{g}_0 + \lambda \gamma^{(1)} + \lambda^2 \gamma^{(2)},\label{def gamma}
\end{align}
where 
\begin{align}
\gamma^{(1)} &  =  \cos\left( \frac{u_0}{\lambda}\right) F_0 \omega^{(1)},\label{def gamma1}
\\ \gamma^{(2)} &  = \sin \left( \frac{u_0}{\lambda}\right)\omega^{(2)} .\label{def gamma2}
\end{align}
The two symmetric 2-forms $\omega^{(1)}$ and $\omega^{(2)}$ are directly defined in the frame $(N_0, e_{\1},e_{\2})$.
\begin{itemize}
\item \textbf{Definition of $\omega^{(1)}$.} The coefficients of $\omega^{(1)}$ in the frame $(N_0, e_{\1},e_{\2})$ are constants and satisfy
\begin{align}
\omega^{(1)}_{N_0 i} & = 0, \label{omega N}
\\ \omega^{(1)}_{\1\1} + \omega^{(1)}_{\2\2}  & = 0, \label{omega 11 + 22}
\\ \left(\omega^{(1)}_{\1\1} \right)^2 + \left(\omega^{(1)}_{\1\2} \right)^2  & = 4.\label{omega backreaction}
\end{align}
\item \textbf{Definition of $\omega^{(2)}$.} The 2-form $\omega^{(2)}$ depends on $\omega^{(1)}$ in the following way: all the coefficients of $\omega^{(2)}$ in the frame $(N_0, e_{\1},e_{\2})$ are set to be zero except $\omega^{(2)}_{N_0\1}$  and $\omega^{(2)}_{N_0\2}$ which are defined by
\begin{align}
\omega^{(2)}_{N_0\mathbf{j}} & = \frac{1}{|\nabla u_0|_{\bar{g}_0}^2} (\dive_{\bar{g}_0} |\nabla u_0|_{\bar{g}_0} F_0\omega^{(1)})_\mathbf{j}  -\frac{1}{\left|\nabla u_0 \right|_{\bar{g}_0} } F_0\omega^{(1)}_{i\mathbf{j}}  \left(  \D_{N_0}N_0^i - \bar{g}_0^{ik} (K_0)_{N_0 k}    \right) . \label{def omega2}
\end{align}
This expression will be justified in the proof of Lemma \ref{lem sigma01}, where we construct the TT-tensor of the conformal method.
\end{itemize}

\begin{remark}
The properties \eqref{omega N} and \eqref{omega 11 + 22} and the symmetry of $\omega^{(1)}$ imply that the tensor $\gamma^{(1)}$ is a linear combination of the tensors
\begin{align*}
e_\1 \otimes e_\1  -  e_\2 \otimes e_\2 \quad \text{and} \quad e_\1 \otimes e_\2  +  e_\2 \otimes e_\1.
\end{align*}
This is the equivalent of the 2 degrees of freedom (or polarization) of the TT gauge in the linearized gravity setting where $g_0$ is replaced by the Minkowskian metric, see Remark \ref{remark TT gauge}. The choice of the polarization of $\omega^{(1)}$, i.e the choice of $\omega^{(1)}_{\1\1}$ and $\omega^{(1)}_{\1\2}$ under the constraint \eqref{omega backreaction}, is the only freedom we have in the definition of the tensors $\bar{F}^{(1)}$, $\bar{F}^{(2,1)}$ and $\bar{F}^{(2,2)}$.
\end{remark}

We define
\begin{align}
\bar{F}^{(1)} & = F_0 \omega^{(1)}.\label{def Fbar 1}
\end{align}
The following lemma summarizes the important properties of $\bar{F}^{(1)}$.
\begin{lem}\label{lem F1}
The symmetric 2-tensor $\bar{F}^{(1)}$ is supported in $B_R$, $\bar{g}_0$-traceless and $P_{0,u}$-tangential. Moreover, \eqref{energie condition theo chap 2} holds and
\begin{align}
\l \bar{F}^{(1)} \r_{H^{N}} \lesssim \e, \label{estim Fbar 1}
\end{align}
with a constant depending only on $\delta$ and $R$.
\end{lem}

\begin{proof}
The support property is implied by the support property of $F_0$. The property \eqref{omega N} implies that $\bar{F}^{(1)}$ is $P_{0,u}$-tangential,  which together with \eqref{omega 11 + 22} implies that $\bar{F}^{(1)}$ is $\bar{g}_0$-traceless since
\begin{align*}
\tr_{\bar{g}_0}\bar{F}^{(1)} = \bar{F}^{(1)}_{N_0N_0} + \bar{F}^{(1)}_{\1\1} + \bar{F}^{(1)}_{\2\2}.
\end{align*}
Moreover, \eqref{omega backreaction} and \eqref{def Fbar 1} imply \eqref{energie condition theo chap 2}. The estimation \eqref{estim Fbar 1} comes from \eqref{estim F0 chap 2} and \eqref{estim u0}, the latter allowing us to estimate the coefficients of $\omega^{(1)}$ in the usual Euclidean coordinates, i.e $\omega^{(1)}_{ij}$.
\end{proof}

The following lemma gives the expansion of the scalar curvature of $\gamma$. This will be used to solve the Hamitonian constraint \eqref{hamiltonian constraint}.
\begin{lem}\label{lem R(gamma)}
We have
\begin{align*}
R(\gamma) & = R^{(0)} +\lambda R^{(1)} + \lambda^2R^{(\geq 2)},
\end{align*}
with
\begin{align}
R^{(0)} &= R(\bar{g}_0)  - |\nabla u_0|_{\Bar{g}_0}^2 F_0^2   -  7  \cos\left(\frac{2u_0}{\lambda}\right) |\nabla u_0|_{\Bar{g}_0}^2 F_0^2   \label{R 0}
\\&\quad +\sin\left(\frac{u_0}{\lambda}\right) |\nabla u_0|_{\Bar{g}_0}\bar{F}^{(1)}_{\ell j} \left( -\bar{g}_0^{ij} \dr_i N_0^\ell + \half N_0 \bar{g}_0^{\ell j} \right),\nonumber
\end{align}
and
\begin{align}
\left| R^{(1)} \right| + \left| R^{(\geq 2)} \right| \lesssim \left| \gamma^{-2}\dr^2\gamma \right| + \left| \gamma^{-3} (\dr\gamma)^2 \right| \label{estim R1 et R2}.
\end{align}
Moreover
\begin{align}
R^{(1)} \simf \cos(\theta) + \sin(2\theta) + \cos(3\theta) .\label{freq R1}
\end{align}
\end{lem}

\begin{proof}
We recall the definition of the scalar curvature
\begin{align*}
R(\gamma) & = \gamma^{ij} \left( \dr_k \Gamma(\gamma)^k_{ij} - \dr_i  \Gamma(\gamma)^k_{jk} +  \Gamma(\gamma)^k_{k\ell} \Gamma(\gamma)^\ell_{ij} -  \Gamma(\gamma)^k_{i\ell} \Gamma(\gamma)^\ell_{jk} \right).
\end{align*}
We start by giving an estimation up to second order in $\lambda$ of the inverse of $\gamma$:
\begin{align}
\gamma^{ij} & = \bar{g}_0^{ij} - \lambda\cos\left(\frac{u_0}{\lambda}\right) (\bar{F}^{(1)})^{ij} \label{inverse gamma}
\\&\quad + \lambda^2 \left( \cos^2 \left(\frac{u_0}{\lambda}\right)  \bar{g}_0^{ik} (\bar{F}^{(1)})^{j\ell} \bar{F}^{(1)}_{k\ell} -  \sin \left( \frac{u_0}{\lambda}\right)(\omega^{(2)})^{ij}  \right) + \GO{\lambda^3},\nonumber
\end{align}
where on the RHS all the inverses are taken with respect to the background metric $\bar{g}_0$. We now expand the Christoffel symbols of $\gamma$, we obtain
\begin{align}
\Gamma(\gamma)^k_{ij} & = \Gamma(\bar{g}_0)^k_{ij} + (\Tilde{\Gamma}^{(0)})^k_{ij} + \lambda (\Gamma^{(1)})^k_{ij}  + \GO{\lambda^2},\label{Christoffel decomposition}
\end{align}
with
\begin{align}
(\Tilde{\Gamma}^{(0)})^k_{ij} & = -\half\sin\left(\frac{u_0}{\lambda}\right)  \bar{g}_0^{k\ell} \left(  \dr_{(i}u_0 \bar{F}^{(1)}_{\ell j)} - \dr_\ell u_0 \bar{F}^{(1)}_{ij}  \right), \label{gamma tilde 0}
\\ (\Gamma^{(1)})^k_{ij}  & =  \half\cos\left(\frac{u_0}{\lambda}\right)  \bar{g}_0^{k\ell} \left(  \dr_{(i}u_0 \omega^{(2)}_{\ell j)} - \dr_\ell u_0 \omega^{(2)}_{ij}  \right)  \label{gamma 1}
\\&\quad +  \frac{1}{4} \sin\left(\frac{2u_0}{\lambda}\right)  (\bar{F}^{(1)})^{k\ell}   \dr_{(i}u_0 \bar{F}^{(1)}_{\ell j)} + \cos\left(\frac{u_0}{\lambda}\right)\bar{Q}^k_{ij} ,\nonumber
\end{align}
where we defined
\begin{align*}
\bar{Q}^k_{ij} & =  \half \bar{g}_0^{k\ell}\left(  \dr_{(i} \bar{F}^{(1)}_{\ell j)} - \dr_\ell \bar{F}^{(1)}_{ij}  \right)  - \half  (\bar{F}^{(1)})^{k\ell}  \left(  \dr_{(i} (\bar{g}_0)_{\ell j)} - \dr_\ell (\bar{g}_0)_{ij}  \right) .
\end{align*}
By using $\bar{F}^{(1)}_{N_0i}=0$ and $\tr_{\bar{g}_0}\bar{F}^{(1)}=0$, we can compute useful contractions of $\bar{Q}^k_{ij}$. We first look at $ \bar{Q}^k_{jk}$ which vanishes thanks to $\tr_{\bar{g}_0}\bar{F}^{(1)}=0$:
\begin{align}
\bar{Q}^k_{jk} & = \half  \left(  \bar{g}_0^{k\ell}   \dr_{j} \bar{F}^{(1)}_{\ell k}   -   (\bar{F}^{(1)})^{k\ell}     \dr_{j} (\bar{g}_0)_{\ell k}  \right) = \half \dr_j \tr_{\bar{g}_0}\bar{F}^{(1)} = 0. \label{prop Q bar 1}
\end{align}
Using in addition $\bar{F}^{(1)}_{N_0i}=0$ we have:
\begin{align}
\bar{g}_0^{ij}(N_0)_k\bar{Q}^k_{ij} & =  \bar{g}_0^{ij}   N_0^\ell  \left(  \dr_{i} \bar{F}^{(1)}_{\ell j} - \half \dr_\ell \bar{F}^{(1)}_{ij}  \right)\nonumber
\\& =  \bar{g}_0^{ij}  \dr_i \bar{F}^{(1)}_{N_0 j}   -   \half  N_0 \tr_{\bar{g}_0}\bar{F}^{(1)} -  \bar{g}_0^{ij}  \bar{F}^{(1)}_{\ell j} \dr_{i} N_0^\ell     +   \half \bar{F}^{(1)}_{\ell j}   N_0 \bar{g}_0^{\ell j}\nonumber
\\& = \bar{F}^{(1)}_{\ell j} \left( -\bar{g}_0^{ij} \dr_i N_0^\ell + \half N_0 \bar{g}_0^{\ell j} \right). \label{prop Q bar 2}
\end{align}

\saut
Since the scalar curvature contains first derivatives of the Christoffel symbols,  $R(\gamma)$ contains \textit{a priori} a $\lambda^{-1}$ contribution from $\dr_\theta \Tilde{\Gamma}^{(0)}$,  but thanks to the properties of $\bar{F}^{(1)}$ we can see that it vanishes.  Indeed we have
\begin{align*}
R^{(-1)} = \bar{g}_0^{ij}\left(  \dr_ku_0 \dr_\theta (\Tilde{\Gamma}^{(0)})^k_{ij} - \dr_i u_0 \dr_\theta (\Tilde{\Gamma}^{(0)})^k_{jk}  \right),
\end{align*}
and $\bar{F}^{(1)}_{N_0 i} = 0$ and $\tr_{\bar{g}_0}\bar{F}^{(1)}=0$ imply that 
\begin{align}
\bar{g}_0^{ij}(\Tilde{\Gamma}^{(0)})^k_{ij} =0\quad \text{and}\quad  (\Tilde{\Gamma}^{(0)})^k_{jk}=0.\label{nice properties gamma tilde}
\end{align}
Let us now look at the $\lambda^0$ terms in $R(\gamma)$.  Using again \eqref{nice properties gamma tilde} we obtain:
\begin{align*}
R^{(0)} - R(\bar{g}_0) & =  \bar{g}_0^{ij}\left(  \dr_ku_0 \dr_\theta (\Gamma^{(1)})^k_{ij} - \dr_i u_0 \dr_\theta (\Gamma^{(1)})^k_{jk}  \right) 
\\&\quad + \bar{g}_0^{ij}  \dr_k  (\Tilde{\Gamma}^{(0)})^k_{ij} -2 \bar{g}_0^{ij} (\Tilde{\Gamma}^{(0)})^k_{i\ell} \Gamma(\bar{g}_0)^\ell_{jk}\nonumber
\\&\quad -  \cos\left(\frac{u_0}{\lambda}\right) (\bar{F}^{(1)})^{ij}  \left(  \dr_ku_0 \dr_\theta (\Tilde{\Gamma}^{(0)})^k_{ij} -\dr_i u_0 \dr_\theta (\tilde{\Gamma}^{(0)})^k_{jk}  \right)  -\bar{g}_0^{ij}  (\Tilde{\Gamma}^{(0)})^k_{i\ell} (\Tilde{\Gamma}^{(0)})^\ell_{jk}\nonumber
\\&= \bar{g}_0^{ij}\left(  \dr_ku_0 \dr_\theta (\Gamma^{(1)})^k_{ij} - \dr_i u_0 \dr_\theta (\Gamma^{(1)})^k_{jk}  \right) \nonumber
\\&\quad + \bar{g}_0^{ij}  \dr_k  (\Tilde{\Gamma}^{(0)})^k_{ij} -2 \bar{g}_0^{ij} (\Tilde{\Gamma}^{(0)})^k_{i\ell} \Gamma(\bar{g}_0)^\ell_{jk}\nonumber
\\&\quad -  \cos\left(\frac{u_0}{\lambda}\right) (\bar{F}^{(1)})^{ij}    \dr_ku_0 \dr_\theta (\Tilde{\Gamma}^{(0)})^k_{ij}  -\bar{g}_0^{ij}  (\Tilde{\Gamma}^{(0)})^k_{i\ell} (\Tilde{\Gamma}^{(0)})^\ell_{jk},\nonumber
\end{align*}
where we used $(\bar{F}^{(1)})^{ij}\dr_i u_0 = 0$. We set 
\begin{align*}
I & \vcentcolon = \bar{g}_0^{ij}\left(  \dr_ku_0 \dr_\theta (\Gamma^{(1)})^k_{ij} - \dr_i u_0 \dr_\theta (\Gamma^{(1)})^k_{jk}  \right) ,
\\ II & \vcentcolon =   \bar{g}_0^{ij}  \dr_k  (\Tilde{\Gamma}^{(0)})^k_{ij} -2 \bar{g}_0^{ij} (\Tilde{\Gamma}^{(0)})^k_{i\ell} \Gamma(\bar{g}_0)^\ell_{jk}  ,
\\ III & \vcentcolon = - \cos\left(\frac{u_0}{\lambda}\right) (\bar{F}^{(1)})^{ij}    \dr_ku_0 \dr_\theta (\Tilde{\Gamma}^{(0)})^k_{ij}  -\bar{g}_0^{ij}  (\Tilde{\Gamma}^{(0)})^k_{i\ell} (\Tilde{\Gamma}^{(0)})^\ell_{jk} ,
\end{align*}
and compute further each of these terms. The term $I$ contains all the contributions from $\Gamma^{(1)}$ and \eqref{gamma 1} implies
\begin{align*}
I & =  \sin\left( \frac{u_0}{\lambda}\right)   |\nabla u_0|_{\bar{g}_0}^2 \left( \omega^{(2)}_{\1\1}  +  \omega^{(2)}_{\2\2}  \right)   -\half\cos\left( \frac{2u_0}{\lambda}\right)   |\nabla u_0|^2_{\bar{g}_0} \left| \bar{F}^{(1)} \right|^2_{\bar{g}_0}
\\&\quad  -\sin\left( \frac{u_0}{\lambda}\right) \bar{g}_0^{ij}\left(  \dr_ku_0 \bar{Q}^k_{ij} - \dr_i u_0 \bar{Q}^k_{jk}  \right)
\\& =  -4\cos\left( \frac{2u_0}{\lambda}\right)   |\nabla u_0|^2_{\bar{g}_0}F_0^2 +\sin\left( \frac{u_0}{\lambda}\right)  |\nabla u_0|_{\bar{g}_0} \bar{F}^{(1)}_{\ell j} \left( -\bar{g}_0^{ij} \dr_i N_0^\ell + \half N_0 \bar{g}_0^{\ell j} \right),
\end{align*}
where we used that the diagonal coefficients of $\omega^{(2)}$ in the frame $(N_0,e_\1,e_\2)$ vanish in addition to \eqref{prop Q bar 1}-\eqref{prop Q bar 2} and \eqref{energie condition theo chap 2}. The term $II$ contains all the linear terms in $\tilde{\Gamma}^{(0)}$. Note that in the notation $\dr_k (\tilde{\Gamma}^{(0)})^k_{ij}$ the derivative does not hit the oscillating part of $\tilde{\Gamma}^{(0)}$ (i.e $\sin\left( \frac{u_0}{\lambda}  \right)$ in \eqref{gamma tilde 0}).  Using \eqref{nice properties gamma tilde} we obtain $ \bar{g}_0^{ij}\dr_k (\tilde{\Gamma}^{(0)})^k_{ij}= - (\tilde{\Gamma}^{(0)})^k_{ij} \dr_k  \bar{g}_0^{ij} $ and
\begin{align*}
II & =  - (\Tilde{\Gamma}^{(0)})^k_{i\ell} \left(   \dr_k \bar{g}_0^{i\ell} + 2 \bar{g}_0^{ij}  \Gamma(\bar{g}_0)^\ell_{jk}\right) =0,
\end{align*}
where we used $(\Tilde{\Gamma}^{(0)})^k_{i\ell}=(\Tilde{\Gamma}^{(0)})^k_{\ell i}$. The term $III$ contains quadratic terms in $\Tilde{\Gamma}^{(0)}$. Using \eqref{gamma tilde 0} and \eqref{energie condition theo chap 2} we obtain
\begin{align*}
III & = -  \cos\left(\frac{u_0}{\lambda}\right) (\bar{F}^{(1)})^{ij}   \dr_ku_0 \dr_\theta (\Tilde{\Gamma}^{(0)})^k_{ij}    -\bar{g}_0^{ij}  (\Tilde{\Gamma}^{(0)})^k_{i\ell} (\Tilde{\Gamma}^{(0)})^\ell_{jk}
\\& = - \half  \cos^2\left(\frac{u_0}{\lambda}\right)  |\nabla u_0|^2_{\bar{g}_0} \left| \bar{F}^{(1)}\right|^2_{\bar{g}_0} + \frac{1}{4} \sin^2\left(\frac{u_0}{\lambda}\right)  |\nabla u_0|^2_{\bar{g}_0} \left| \bar{F}^{(1)}\right|^2_{\bar{g}_0}
\\& = - |\nabla u_0|^2_{\bar{g}_0} F_0^2  - 3\cos\left(\frac{2u_0}{\lambda}\right) |\nabla u_0|^2_{\bar{g}_0} F_0^2.
\end{align*}
We conclude the proof of \eqref{R 0} by adding $I$ and $III$.

\saut
The proof of \eqref{freq R1} comes from the schematic formula
\begin{align*}
R(\gamma) = \gamma^{-2}\dr^2\gamma + \gamma^{-3}(\dr\gamma)^2,
\end{align*}
which gives schematically
\begin{align}
R^{(1)} & = \left(   (\gamma^{-1})^{(2)}(\gamma^{-1})^{(0)} + \left( (\gamma^{-1})^{(1)} \right)^2 \right)(\dr^2\gamma)^{(-1)}  +  (\gamma^{-1})^{(1)} (\gamma^{-1})^{(0)}(\dr^2\gamma)^{(0)} \label{R1 preuve}
\\&\quad+ \left( (\gamma^{-1})^{(0)}\right)^2 (\dr^2 \gamma)^{(1)}  + (\gamma^{-1})^{(1)}\left( (\gamma^{-1})^{(0)}(\dr\gamma)^{(0)}\right)^2   + \left( (\gamma^{-1})^{(0)}\right)^3(\dr\gamma)^{(1)}(\gamma^{-1})^{(0)},\nonumber
\end{align}
where we recall that our high-frequency notations introduced in Section \ref{section HF notations} give for instance 
\begin{align*}
 (\dr^2 \gamma)^{(1)} & = \cos\left( \frac{u_0}{\lambda} \right) \dr^2 \bar{F}^{(1)} + 2 \cos\left( \frac{u_0}{\lambda} \right) \dr\omega^{(2)},
\end{align*}
and
\begin{align*}
(\gamma^{-1})^{(1)} & = - \cos\left( \frac{u_0}{\lambda}\right) \bar{F}^{(1)}.
\end{align*}
Using \eqref{def gamma} and \eqref{inverse gamma} we obtain
\begin{align*}
(\gamma^{-1})^{(0)} &\simf 1, 
\\ (\gamma^{-1})^{(1)} &\simf \cos(\theta) ,  
\\ (\gamma^{-1})^{(2)} &\simf 1 + \sin(\theta) + \cos(2\theta),
\end{align*}
and 
\begin{align*}
(\dr\gamma)^{(0)},\; (\dr^2\gamma)^{(0)} &\simf 1 + \sin(\theta)  , 
\\ (\dr\gamma)^{(1)} ,\; (\dr^2\gamma)^{(-1)} ,\; (\dr^2\gamma)^{(1)} &\simf \cos(\theta).
\end{align*}
Using these oscillating behaviours we can prove by a direct computation that $R^{(1)}$ defined by \eqref{R1 preuve} satisfies \eqref{freq R1}. This concludes the proof of the lemma.
\end{proof}

\begin{remark}
Note that we denote by $\gamma^{-k}$ any product of $k$ coefficients of the inverse metric $\gamma^{-1}$. This also applies to the background inverse metric $\bar{g}_0^{-1}$.
\end{remark}

\subsection{Expansion of differential operators}\label{section exp diff op}

The equations \eqref{hamiltonian constraint} and \eqref{momentum constraint} involve differential operators depending on the metric $\gamma$. When they are applied to high-frequency quantities, we need to take into account the expansion \eqref{def gamma} affecting the coefficients of the operators as well as the expansion of the quantities themselves.  In this section we compute such expansions for all the operators involved, that is the Laplace-Beltrami operator, the divergence operator, the conformal Killing operator and the conformal Laplacian. 

\saut
In terms of notation, if $P_\gamma$ is a differential operator acting on tensors of any type and whose coefficients depends on $\gamma$, we can formally obtain an expansion of $P_\gamma (T)$ of the form 
\begin{align*}
P_\gamma (T) & = \sum_k \lambda^k P^{[k]}_\gamma (T).
\end{align*}
where the previous sum has finite support. The bracket notation $[k]$ thus plays the same role for differential operators as the parenthesis notation $(i)$ introduced in Section \ref{section HF notations} for tensors. 

\saut
Moreover we can mix the two cases, i.e apply differential operators $P_\gamma$ depending on $\gamma$ to oscillating tensors $T\left(\frac{u_0}{\lambda}\right)$. The expansion of $P_\gamma (T)$ then depends on the order of $T\longmapsto P_\gamma (T)$ and also $\gamma\longmapsto P_\gamma (T)$. By the order of $\gamma\longmapsto P_\gamma (T)$, we simply mean the top derivative of $\gamma$ appearing in the coefficients of $P_\gamma$.  Since both oscillating and non-oscillating terms appear in the expansions for the parameters and the unknowns of the conformal method (see Section \ref{section HF parameters unknowns}), we make a difference between the expansions for $P_\gamma(T)$ when $T$ is non-oscillating and when $T$ is oscillating. Usual capital letters are used in the first case and bold capital letters in the second case, i.e
\begin{align*}
P_\gamma (T) & = \sum_k \lambda^k P^{[k]}_\gamma (T) \quad \text{and} \quad P_\gamma \left(T\left(\frac{u_0}{\lambda}\right)\right)  = \sum_k \lambda^k \mathbf{P}^{[k]}_\gamma (T),
\end{align*}
where the support of the two previous finite sums are \textit{a priori} different, depending on $T$ and $P_\gamma$. This explains the difference between Lemmas \ref{lem conf laplace non osc} and \ref{lem conf laplace osc} below.

\subsubsection{The Laplace-Beltrami operator}

The only differential operator in the hamiltonian constraint \eqref{hamiltonian constraint} is the Laplace-Beltrami operator associated to $\gamma$. If $h$ is a generic Riemannian metric on $\Sigma_0$, we define its Laplace-Beltrami operator by
\begin{align*}
\Delta_h f  = h^{ij}\left( \dr_i\dr_j f - \Gamma(h)_{ij}^k \dr_k f\right),
\end{align*}
for $f$ a scalar function. Note that $\Delta_e$ is the usual Laplacian operator on $\R^3$ and is denoted by $\Delta$ in this article.  Since $\gamma=\GO{\lambda^0}$ and $\dr\gamma=\GO{\lambda^0}$, if $f$ does not admit a high-frequency expansion then
\begin{align*}
 \Delta_\gamma f = \GO{\lambda^0}.
\end{align*}
If $f$ admit a high-frequency expansion we have the following lemma.

\begin{lem}\label{lem expansion laplace}
For $f\left( \frac{u_0}{\lambda}\right)$ an oscillating scalar function we have
\begin{align*}
\Delta_\gamma \left( f\left( \frac{u_0}{\lambda}\right)\right) &   = \frac{1}{\lambda^2}\mathbf{H}^{[-2]}(f)+\frac{1}{\lambda} \mathbf{H}^{[-1]}(f)  +  \mathbf{H}^{[\geq 0]}(f) ,
\end{align*}
with
\begin{align}
\mathbf{H}^{[-2]}(f) & = |\nabla u_0|_{\bar{g}_0}^2\dr_\theta^2f ,  \label{H-2 a}
\\\mathbf{H}^{[-1]}(f) & =  2\Bar{g}_0^{ij}\dr_i u_0 \dr_j\dr_\theta f + (\Tilde{\Delta}_{\Bar{g}_0}u_0)\dr_\theta f -\bar{g}_0^{ij}\Gamma(\bar{g}_0)_{ij}^k\dr_k u_0 \dr_\theta f  ,\label{H-1 a}
\end{align}
and
\begin{align}
\left| \mathbf{H}^{[\geq 0]}(f) \right| \lesssim \left| \gamma^{-1}\dr^2 f \right| + \left| \gamma^{-2}\dr\gamma \dr f \right|    .    \label{H0}
\end{align}
\end{lem}

\begin{proof}
We start with the definition of the Laplace-Beltrami operator :
\begin{equation*}
\Delta_\gamma \left( \ffi\left( \frac{u_0}{\lambda}\right)\right) = \gamma^{ij}\dr_i\dr_j\left(\ffi\left( \frac{u_0}{\lambda}\right) \right) - \gamma^{ij}\Gamma(\gamma)_{ij}^k\dr_k \left(\ffi\left( \frac{u_0}{\lambda}\right) \right).
\end{equation*}
Using the expansion of the inverse of $\gamma$ and $\bar{F}^{(1)}_{N_0 i}=0$ we have
\begin{align*}
\gamma^{ij}\dr_i\dr_j\left(\ffi\left( \frac{u_0}{\lambda}\right) \right)&=\frac{1}{\lambda^2}|\nabla u_0|_{\gamma}^2\dr_\theta^2\ffi
+\frac{1}{\lambda}\left( 2\gamma^{ij}\dr_i u_0 \dr_j\dr_\theta\ffi + (\Tilde{\Delta}_{\gamma}u_0)\dr_\theta\ffi \right)   + \GO{\lambda^0}
\\&=\frac{1}{\lambda^2}|\nabla u_0|_{\bar{g}_0}^2\dr_\theta^2\ffi + \frac{1}{\lambda}\left( 2\Bar{g}_0^{ij}\dr_i u_0 \dr_j\dr_\theta\ffi  + (\Tilde{\Delta}_{\Bar{g}_0}u_0)\dr_\theta\ffi  \right)  + \GO{\lambda^0},
\end{align*}
where $\tilde{\Delta}_h = h^{ij}\dr_i \dr_j$. Moreover from the decomposition of the Christoffel symbols \eqref{Christoffel decomposition} and \eqref{nice properties gamma tilde} we obtain
\begin{align*}
\gamma^{ij}\Gamma(\gamma)_{ij}^k\dr_k \left(\ffi\left( \frac{u_0}{\lambda}\right) \right) & =\frac{1}{\lambda}\bar{g}_0^{ij}\Gamma(\bar{g}_0)_{ij}^k\dr_k u_0 \dr_\theta\ffi   + \GO{\lambda}.
\end{align*}
The estimate \eqref{H0} simply comes from the definition of $\Delta_\gamma$.
\end{proof}

\subsubsection{The divergence operator}

The divergence operator appears in the momentum constraint \eqref{momentum constraint} but also in the definition of the parameter $\sigma$, which in particular needs to be a divergence free tensor for the metric $\gamma$. Recall that if $h$ is a Riemannian metric on $\Sigma_0$ and if $\D^{(h)}$ denotes the covariant derivative associated, then $\dive_h A = h^{k\ell}\D^{(h)}_k A_{\ell}$ and $(\dive_h B)_i = h^{k\ell}\D^{(h)}_k B_{\ell i}$ for $A$ a 1-tensor and $B$ a 2-tensor. The divergence operator only depends on first derivatives of $\gamma$ which are $\GO{\lambda^0}$ so we only need an expansion when the tensor on which it acts is itself oscillating.

\begin{lem}\label{lem divergence gamma}
For $A_{ij}\left(\frac{u_0}{\lambda}\right)$ an oscillating symmetric 2-tensor  we have 
\begin{align*}
\dive_\gamma A\left(\frac{u_0}{\lambda}\right)_\ell & = \frac{1}{\lambda}\mathbf{d}^{[-1]}_\ell(A) +   \mathbf{d}^{[0]}_\ell(A) + \lambda \mathbf{d}^{[\geq 1]}_\ell(A) ,
\end{align*}
with
\begin{align}
\mathbf{d}^{[-1]}_\ell(A) & = -  |\nabla u_0|_{\Bar{g}_0} \dr_\theta A_{N_0\ell} ,\label{d-1}
\\ \mathbf{d}^{[0]}_\ell(A) & =\dive_{\Bar{g}_0} A_\ell - (\Bar{g}_0)^{ij} (\Tilde{\Gamma}^{(0)})^a_{i\ell} A_{aj},\label{d0}
\end{align}
and
\begin{align}
\left| \mathbf{d}^{[\geq 1]}_\ell(A) \right| \lesssim \left| \gamma^{-1}\dr A \right| + \left| \gamma^{-2}\dr\gamma  A \right|    . \label{d1}
\end{align}
Moreover
\begin{align}
\mathbf{d}^{[0]}(A) & \simf (1+\sin(\theta))A, \label{d0 oscillation}
\\ \mathbf{d}^{[1]}(A) & \simf \sin(\theta) \dr_\theta A  +   (\cos(\theta) + \sin(2\theta)) A. \label{d1 oscillation}
\end{align}
\end{lem}

\begin{proof}
We start by using the expansion of the Christoffel's symbols from \eqref{Christoffel decomposition}:
\begin{align*}
\D^{(\gamma)}_i A\left(\frac{u_0}{\lambda}\right)_{j\ell} & =  \frac{1}{\lambda} \dr_i u_0 \dr_\theta A_{j\ell} + \Db_i A_{j\ell} - (\Tilde{\Gamma}^{(0)})^a_{i(j}  A_{a\ell)} - \lambda (\Gamma^{(1)})^a_{i(j} A_{a\ell)} + \GO{\lambda^2}.
\end{align*}
We now use the expansion of the inverse of $\gamma$:
\begin{align}
\dive_\gamma A\left(\frac{u_0}{\lambda}\right)_\ell  & = \gamma^{ij} \D^{(\gamma)}_i A\left(\frac{u_0}{\lambda}\right)_{j\ell} \nonumber
\\& = - \frac{1}{\lambda} |\nabla u_0|_{\Bar{g}_0} \dr_\theta A_{N_0\ell} + \dive_{\Bar{g}_0} A_\ell - \bar{g}_0^{ij}(\Tilde{\Gamma}^{(0)})^a_{i(j}  A_{a\ell)} - \cos\left(\frac{u_0}{\lambda}\right) (\bar{F}^{(1)})^{ij}\dr_i u_0 \dr_\theta A_{j\ell}\label{dive A calcul preuve}
\\&\quad + \lambda \left[     \cos^2 \left(\frac{u_0}{\lambda}\right)  \bar{g}_0^{ik} (\bar{F}^{(1)})^{ja} \bar{F}^{(1)}_{ka}\dr_i u_0 \dr_\theta A_{j\ell}   -\sin\left( \frac{u_0}{\lambda}\right) (\omega^{(2)})^{ij}\dr_i u_0 \dr_\theta A_{j\ell}        \right.\nonumber
\\&\qquad\qquad \qquad       \left.                 - \cos\left(\frac{u_0}{\lambda}\right) (\bar{F}^{(1)})^{ij} \left( \Db_i A_{j\ell} - (\Tilde{\Gamma}^{(0)})^a_{i(j}  A_{a\ell)} \right) -\bar{g}_0^{ij}  (\Gamma^{(1)})^a_{i(j} A_{a\ell)}             \right]\nonumber
\\&\quad + \GO{\lambda^2}.\nonumber
\end{align}
This concludes the proof, using $\bar{F}^{(1)}_{N_0i}=0$ and $\Bar{g}_0^{ij}(\Tilde{\Gamma}^{(0)})^a_{ij}=0$. The estimate \eqref{d1} just comes from the definition of the divergence. The oscillating behaviours \eqref{d0 oscillation} and \eqref{d1 oscillation} can be directly read on \eqref{dive A calcul preuve}.
\end{proof}

\subsubsection{The conformal Killing operator}

In order to compute $L_\gamma W$ in the ansatz for $K$ \eqref{K}, we need to expand the conformal Killing operator defined by
\begin{align}
(L_h A)_{ij} = \D^{(h)}_{(i} A_{j)} - \frac{2}{3}(\dive_h A)h_{ij}, \label{def L W}
\end{align}
where $A$ is a 1-tensor. Note that $L_h A$ is a symmetric 2-tensor traceless with respect to $h$.  As for the divergence operator, $L_\gamma$ only depends on first derivatives of $\gamma$ so only an expansion of $L_\gamma W$ with $W$ oscillating is necessary.

\begin{lem}\label{lem conformal killing}
Let $W\left(\frac{u_0}{\lambda} \right)$ be an oscillating 1-form. We have
\begin{equation*}
 L_\gamma\left( W\left( \frac{u_0}{\lambda} \right)  \right)_{ij} = \frac{1}{\lambda} \mathbf{K}_{ij}^{[-1]}(W) +\mathbf{K}_{ij}^{[0]}(W) +\lambda \mathbf{K}_{ij}^{[\geq 1]}(W),
\end{equation*}
where
\begin{align}
\mathbf{K}_{ij}^{[-1]}(W) & = \dr_{(i} u_0 \dr_\theta W_{j)} + \frac{2}{3}|\nabla u_0|_{\bar{g}_0}(\Bar{g}_0)_{ij} \dr_\theta W_{N_0},\label{K-1}
\\ \mathbf{K}_{ij}^{[0]}(W) & = \Db_{(i}W_{j)}  -  2 W_k (\Tilde{\Gamma}^{(0)})^k_{ij}  + \frac{2}{3}\cos\left( \frac{u_0}{\lambda}\right)|\nabla u_0|_{\bar{g}_0}\Bar{F}^{(1)}_{ij} \dr_\theta W_{N_0} - \frac{2}{3} \dive_{\bar{g}_0}W (\Bar{g}_0)_{ij},\label{K 0}
\end{align}
and
\begin{align}
\left|  \mathbf{K}_{ij}^{[\geq 1]}(W)  \right| \lesssim  \left| \dr W \right|   +   \left| \gamma^{-1}\dr \gamma W \right|  .
\end{align}
The following hold:
\begin{align}
\tr_{\Bar{g}_0}\mathbf{K}^{[-1]}(W) & = 0,\label{trace K-1}
\\ \tr_{\Bar{g}_0}\mathbf{K}^{[0]}(W)  & = \cos\left( \frac{u_0}{\lambda}\right)  \tr_{\Bar{F}^{(1)}}\mathbf{K}^{[-1]}(W) .\label{trace K0}
\end{align}
The following oscillating behaviour holds
\begin{align}
\mathbf{K}^{[0]}(W) & \simf (1 + \sin(\theta))W + \cos(\theta) \dr_\theta W           . \label{K0 oscillation}
\end{align}
\end{lem}

\begin{proof}
We start with the covariant derivative of the oscillating 1-form $W^{(\cdot)}\left( \frac{u_0}{\lambda} \right)$:
\begin{align*}
\D^{(\gamma)}_i \left(W\left( \frac{u_0}{\lambda}\right) \right)_j & = \frac{1}{\lambda}\dr_i u_0 \dr_\theta W_j + \Db_iW_j\left( \frac{u_0}{\lambda}\right)- W_k\left(   (\Tilde{\Gamma}^{(0)})^k_{ij}+\lambda (\Gamma^{(1)})^k_{ij} + \lambda^2 (\Gamma^{(2)})^k_{ij}+\cdots \right),
\end{align*}
which also gives us the divergence 
\begin{align*}
\dive_\gamma\left(W\left( \frac{u_0}{\lambda}\right) \right) & = \gamma^{ij}\D^{(\gamma)}_i \left(W\left( \frac{u_0}{\lambda}\right) \right)_j
\\& = \frac{1}{\lambda}(\Bar{g}_0)^{k\ell}\dr_k u_0 \dr_\theta W_\ell -\cos\left(\frac{u_0}{\lambda}\right)(\Bar{F}^{(1)})^{k\ell}\dr_k u_0 \dr_\theta W_\ell 
\\&\quad + \dive_{\bar{g}_0}W  - (\Bar{g}_0)^{k\ell}(\Tilde{\Gamma}^{(0)})^a_{k\ell}W_a   +\GO{\lambda}.
\end{align*}
We conclude the proof of \eqref{K-1}-\eqref{K 0} with
\begin{align*}
L_\gamma\left( W\left( \frac{u_0}{\lambda} \right)  \right)_{ij} & = \D^{(\gamma)}_{(i} \left(W\left( \frac{u_0}{\lambda}\right) \right)_{j)} - \frac{2}{3}\dive_\gamma\left(W\left( \frac{u_0}{\lambda}\right) \right)\gamma_{ij},
\end{align*}
and $\bar{F}^{(1)}_{N_0i}=0$ and $(\Bar{g}_0)^{ij}(\Tilde{\Gamma}^{(0)})^a_{ij}=0$.  The trace identities \eqref{trace K-1}-\eqref{trace K0} and the oscillating behaviour \eqref{K0 oscillation} follows directly from \eqref{K-1}-\eqref{K 0}.
\end{proof}

\begin{remark}
In the previous lemma we used the notation $\tr_{\bar{F}^{(1)}} \mathbf{K}^{[-1]}(W)$ to denote \[(\bar{F}^{(1)})^{ij}\mathbf{K}^{[-1]}_{ij}(W) = \bar{g}_0^{ik}\bar{g}_0^{j\ell} \bar{F}^{(1)}_{k\ell} \mathbf{K}^{[-1]}_{ij}(W),  \] even though $\bar{F}^{(1)}$ is not a Lorentzian metric.
\end{remark}

\subsubsection{The conformal Laplacian}\label{section conformal laplacian}

The conformal Laplacian associated to a Riemannian metric $h$ is the operator $\dive_h L_h $ acting on 1-tensor.  It is the only operator considered in this article depending on $\dr^2\gamma$, i.e second derivatives of the metric $\gamma$. Indeed $L_\gamma$ contains the Christoffel symbols of $\gamma$ through the covariant derivative, and they are differenciated once by the divergence. Since $\dr^2\gamma = \GO{\lambda^{-1}}$, this a major difference with the Laplace-Beltrami operator or the divergence operator from the point of view of expanding quantities in powers of $\lambda$. Indeed, this implies that even if the 1-form $\beta$ is not oscillating, the quantity $\dive_\gamma L_\gamma \beta$ still loses one power of $\lambda$. If $\beta$ is oscillating, then $\dive_\gamma L_\gamma \beta$ loses two powers of $\lambda$ since the conformal Laplacian is a second order operator. This explains the two following lemmas.

\begin{lem}\label{lem conf laplace non osc}
Let $\beta$ be a 1-form.  We have
\begin{align*}
\dive_\gamma L_\gamma \beta & = \frac{1}{\lambda} M^{[-1]}(\beta) + M^{[\geq 0]}(\beta),
\end{align*}
with
\begin{align}
M^{[-1]}_\ell(\beta) & =  |\nabla u_0|^2_{\bar{g}_0}   \cos\left(\frac{u_0}{\lambda}\right)  \bar{g}_0^{ij}   \bar{F}^{(1)}_{i \ell}   \beta_j , \label{m-1}
\end{align}
and
\begin{align}
\left| M^{[\geq 0]}(\beta) - \dive_e L_e \beta  \right| & \lesssim \left| \gamma^{-1} - e^{-1}\right| |\dr^2\beta| \label{estimation m0} + \left|\gamma^{-2}\dr\gamma\right| \left| 1 + \gamma^{-1}\gamma\right|   |\dr\beta| 
\\&\quad + \left|\gamma^{-3}(\dr\gamma)^2 + \gamma^{-2}(\dr^2\gamma)^{(\geq 0)} \right| \left| 1 + \gamma^{-1}\gamma\right|   |\beta|\nonumber.
\end{align}
\end{lem}

\begin{proof}
Let $\bar{g}$ be any Riemannian metric on $\R^3$, we have
\begin{align}
\dive_{\bar{g}} L_{\bar{g}} \beta_\ell & = \bar{g}^{ij} \left(  \dr_i \dr_{(j} \beta_{\ell)} - \frac{2}{3} \dr_\ell \dr_i \beta_j  \right)  + \beta_k \left(  -2\bar{g}^{ij}\dr_i \Gamma(\bar{g})^k_{j\ell} + \frac{2}{3}\dr_\ell\left( \bar{g}^{ij}\Gamma(\bar{g})^k_{ij}  \right)   \right)\label{laplacien conforme formule exacte}
\\&\quad - \bar{g}^{ij}\left( 2\Gamma(\bar{g})^k_{j\ell}\dr_i \beta_k + \frac{2}{3}\dive_{\bar{g}}\beta \dr_i\bar{g}_{j\ell} \right) + \frac{2}{3}\left(  -\dr_i\beta_j \dr_\ell \bar{g}^{ij} + \bar{g}^{ij}\Gamma(\bar{g})^k_{ij} \dr_\ell\beta_k   \right) \nonumber
\\&\quad - \bar{g}^{ij}\Gamma(\bar{g})^k_{i(j}L_{\bar{g}}\beta_{k\ell)}.\nonumber
\end{align}
In the case of the high-frequency metric $\gamma$, the only terms loosing $\frac{1}{\lambda}$, i.e contributing to $M^{[-1]}(\beta)$,  are the $\gamma^{-2}\dr^2\gamma \beta$.  More precisely, it concerns terms involving $\dr \tilde{\Gamma}^{(0)}$, and since $\bar{g}_0^{ij}(\tilde{\Gamma}^{(0)})^k_{ij}=0$, the only contribution is
\begin{align*}
M^{[-1]}_\ell(\beta) & = -2 \beta_k\bar{g}^{ij}\dr_i u_0 \dr_\theta (\tilde{\Gamma}^{(0)})^k_{j\ell} = |\nabla u_0|^2_{\bar{g}_0}   \cos\left(\frac{u_0}{\lambda}\right)  \bar{g}_0^{ij}   \bar{F}^{(1)}_{i \ell}   \beta_j ,
\end{align*}
where we used \eqref{gamma tilde 0} and $\bar{F}^{(1)}_{N_0 i}=0$. From \eqref{laplacien conforme formule exacte} we also get
\begin{align*}
\dive_{e} L_{e} \beta_\ell & = e^{ij} \left(  \dr_i \dr_{(j} \beta_{\ell)} - \frac{2}{3} \dr_\ell \dr_i \beta_j  \right),
\end{align*}
which concludes the proof of \eqref{estimation m0}.
\end{proof}

The estimation \eqref{estimation m0} will allow us to invert the operator $M^{[\geq 0]}$, since we know how to invert $\dive_e L_e$ (see Proposition \ref{prop Delta et dive e L e}) and since we gain a smallness constant $\e$ in front of $\dr^2\beta$ thanks to \eqref{estim g0 K0} and \eqref{inverse gamma}. We now state the final lemma of this section,  which will allow us to construct high-frequency solutions of the momentum constraint. The proof is left to the reader since it follows from Lemmas \ref{lem divergence gamma} and \ref{lem conformal killing}.

\begin{lem}\label{lem conf laplace osc}
Let $W\left(\frac{u_0}{\lambda} \right)$ be an oscillating 1-form. We have
\begin{align*}
\dive_\gamma L_\gamma\left( W\left( \frac{u_0}{\lambda} \right)  \right)_\ell & = \frac{1}{\lambda^2} \mathbf{M}^{[-2]}_\ell(W) + \frac{1}{\lambda}\mathbf{M}^{[-1]}_\ell(W)+ \mathbf{M}^{[\geq 0]}_\ell(W) ,
\end{align*}
with
\begin{align}
 \mathbf{M}^{[-2]}_\ell(W) & = \mathbf{d}^{[-1]}_\ell\left( \mathbf{K}^{[-1]} (W)  \right),\label{M-2}
 \\ \mathbf{M}^{[-1]}_\ell(W) & = \mathbf{d}^{[-1]}_\ell\left( \mathbf{K}^{[0]} (W)  \right) +  \mathbf{d}^{[0]}_\ell\left( \mathbf{K}^{[-1]} (W)  \right).\label{M-1}
\end{align}
The following oscillating behaviour holds
\begin{align}
\mathbf{M}^{[-1]}(W) & \sim    \cos(\theta) W   +  \cos(\theta) \dr^2_\theta W    + (1+\sin(\theta))   \dr_\theta W.\label{M-1 oscillation}
\end{align}
\end{lem}

From \eqref{d-1}, \eqref{K-1} and \eqref{M-2} we obtain
\begin{align}
\mathbf{M}^{[-2]}_\ell(W) & = |\nabla u_0|^2_{\bar{g}_0} \left(  \dr^2_\theta W_\ell + \frac{1}{3}(N_0)_\ell \dr^2_\theta W_{N_0}   \right).\label{M-2 expression}
\end{align}
In order to solve the momentum constraint, we need to invert this operator. This is done in the following simple lemma.

\begin{lem}\label{inversion M-2}
If $\alpha$ and $\beta$ are 1-forms such that
\begin{align}
\alpha_\ell + \frac{1}{3} (N_0)_\ell \alpha_{N_0} = \beta_\ell, \label{eq alpha}
\end{align}
then
\begin{align*}
\alpha_\ell = \beta_\ell - \frac{1}{4}(N_0)_\ell \beta_{N_0}.
\end{align*}
\end{lem}

\begin{proof}
We contract \eqref{eq alpha} with the vector field $N_0$ to obtain
\begin{align*}
\frac{4}{3}\alpha_{N_0} = \beta_{N_0}.
\end{align*}
Inserting this into \eqref{eq alpha} concludes the proof.
\end{proof}

\section{Approximate solution to the constraint equations}\label{section approximate solution}

In this section, we construct an approximate solution to the constraint equations \eqref{hamiltonian constraint} and \eqref{momentum constraint} by solving the $\lambda^0$ and $\lambda^1$ Hamiltonian levels and the $\lambda^0$ momentum level.  In the process we will define the parameter $\tau$ and in Sections \ref{section almost TT} and \ref{section TT} we construct the parameter $\sigma$.

\subsection{The approximate Hamiltonian constraint}

We solve the first two levels of the Hamiltonian constraint by choosing $\ffi^{(2)}$ and $\ffi^{(3)}$.  At those levels, \eqref{hamiltonian constraint} decouples from \eqref{momentum constraint} since we replace $(\sigma + L_\gamma W)^{(\leq 1)}$ by $\left( K - \frac{1}{3}\tau \gamma \right)^{(\leq 1)}$, where $K^{(\leq 1)}$ and $\tau^{(\leq 1)}$ will be defined along the way.

\subsubsection{The $\lambda^0$ Hamiltonian level} \label{section ham 0}

We want $\ffi^{(2)}$ to solve the $\lambda^0$ Hamiltonian level. Using the expansion of the Laplace-Beltrami operator $\Delta_\gamma$ (see Lemma \ref{lem expansion laplace}), this is equivalent to the following equation:
\begin{align}
8 |\nabla u_0|_{\Bar{g}_0}^2 \dr^2_\theta \ffi^{(2)} =  R^{(0)} + \frac{2}{3}\left( \tau^{(0)}\right)^2 - \left|K^{(0)} -\frac{1}{3}\tau^{(0)}\bar{g}_0  \right|^2_{\bar{g}_0}.\label{eq ffi 2 first}
\end{align}
where $K^{(0)}$ is defined by
\begin{align}
K^{(0)} & = K_0 + \half \sin\left(\frac{u_0}{\lambda}\right) |\nabla u_0|_{\bar{g}_0} \bar{F}^{(1)},\label{def K0}
\end{align}
following the discussion of Section \ref{section strategy TT tensor}. As the parameter $\tau$ is the trace with respect to $\gamma$ of $K$, we have
\begin{align}
\tau^{(0)} & = \tr_{\bar{g}_0}K^{(0)}\nonumber
\\&= \tr_{\bar{g}_0} K_0 + \half \sin\left(\frac{u_0}{\lambda}\right) |\nabla u_0|_{\bar{g}_0} \tr_{\bar{g}_0} \bar{F}^{(1)} \nonumber
\\& = 0,\label{tau0=0}
\end{align}
where we used \eqref{maximal} and Lemma \ref{lem F1}. Therefore, \eqref{eq ffi 2 first} rewrites
\begin{align}
8 |\nabla u_0|_{\Bar{g}_0}^2 \dr^2_\theta \ffi^{(2)} =  R^{(0)}  - \left|K^{(0)}   \right|^2_{\bar{g}_0}.\label{eq ffi 2}
\end{align}
The LHS of \eqref{eq ffi 2} being a derivative with respect to $\theta$, we need the RHS to be purely oscillating.  Since this shows how the creation of non-oscillating terms is absorbed by the background constraint equations, we state this in a separate lemma.

\begin{lem}
We have
\begin{align*}
R^{(0)}  - \left|K^{(0)}   \right|^2_{\bar{g}_0} & = \sin\left(\frac{u_0}{\lambda}\right) |\nabla u_0|_{\Bar{g}_0}   \bar{F}^{(1)}_{\ell j} \left( -\bar{g}_0^{ij} \dr_i N_0^\ell + \half N_0 \bar{g}_0^{\ell j} - K_0^{\ell j} \right)   -6 \cos\left(  \frac{2u_0}{\lambda}  \right) |\nabla u_0|_{\Bar{g}_0}^2 F_0^2 .
\end{align*}
\end{lem}

\begin{proof}
Let us first expand $\left|K^{(0)}   \right|^2_{\bar{g}_0}$ using \eqref{def K0} and \eqref{energie condition theo chap 2}:
\begin{align*}
\left| K^{(0)} \right|^2_{\Bar{g}_0} & =  |K_0|^2_{\Bar{g}_0} +  |\nabla u_0|_{\Bar{g}_0}^2 F_0^2 -\cos\left(  \frac{2u_0}{\lambda}  \right) |\nabla u_0|_{\Bar{g}_0}^2 F_0^2 + \sin\left( \frac{u_0}{\lambda}\right) |\nabla u_0|_{\Bar{g}_0} \left|\bar{F}^{(1)}\cdot K_0 \right|_{\bar{g}_0}   .
\end{align*}
We use now \eqref{R 0} to compute the full RHS of \eqref{eq ffi 2}:
\begin{align*}
R^{(0)}  - \left|K^{(0)}   \right|^2_{\bar{g}_0} & = R(\bar{g}_0) - |K_0|^2_{\Bar{g}_0} - 2|\nabla u_0|_{\Bar{g}_0}^2 F_0^2   
\\&\quad +\sin\left(\frac{u_0}{\lambda}\right) |\nabla u_0|_{\Bar{g}_0}  \left( \bar{F}^{(1)}_{\ell j} \left( -\bar{g}_0^{ij} \dr_i N_0^\ell + \half N_0 \bar{g}_0^{\ell j} \right)  - \left|\bar{F}^{(1)}\cdot K_0 \right|_{\bar{g}_0}   \right)
\\&\quad  -6 \cos\left(  \frac{2u_0}{\lambda}  \right) |\nabla u_0|_{\Bar{g}_0}^2 F_0^2 .
\end{align*}
Using the background Hamiltonian constraint \eqref{hamiltonian background} to cancel the non-oscillating term, we are left with
\begin{align*}
R^{(0)}  - \left|K^{(0)}   \right|^2_{\bar{g}_0} & = \sin\left(\frac{u_0}{\lambda}\right) |\nabla u_0|_{\Bar{g}_0}  \left( \bar{F}^{(1)}_{\ell j} \left( -\bar{g}_0^{ij} \dr_i N_0^\ell + \half N_0 \bar{g}_0^{\ell j} \right)  - \left|\bar{F}^{(1)}\cdot K_0 \right|_{\bar{g}_0}   \right)
\\&\quad  -6 \cos\left(  \frac{2u_0}{\lambda}  \right) |\nabla u_0|_{\Bar{g}_0}^2 F_0^2,
\end{align*}
which concludes the proof.
\end{proof}

This lemma allows us to solve \eqref{eq ffi 2} by simply setting
\begin{align}
\ffi^{(2)} &  = \frac{1}{8|\nabla u_0|_{\bar{g}_0}}\sin\left(\frac{u_0}{\lambda}\right)  \bar{F}^{(1)}_{\ell j} \left(  \bar{g}_0^{ij} \dr_i N_0^\ell - \half N_0 \bar{g}_0^{\ell j} + K_0^{\ell j} \right)   +\frac{3}{16} \cos\left(\frac{2u_0}{\lambda}\right) F_0^2 . \label{def ffi 2}
\end{align}
Now that $\ffi^{(2)}$ is defined, we can set
\begin{align}
\bar{F}^{(2,1)} & = \frac{\bar{F}^{(1)}_{\ell j}}{2|\nabla u_0|_{\bar{g}_0}}  \left( \bar{g}_0^{ij} \dr_i N_0^\ell - \half N_0 \bar{g}_0^{\ell j} + K_0^{\ell j}  \right) \bar{g}_0 + \omega^{(2)},\label{def F21}
\\ \bar{F}^{(2,2)} & = \frac{3}{4}  F_0^2\bar{g}_0 .\label{def F22}
\end{align}
These definitions are consistent with \eqref{g bar} and the following lemma summarizes the properties of $\bar{F}^{(2,1)}$ and $\bar{F}^{(2,2)}$:

\begin{lem}
The symmetric 2-tensors $\bar{F}^{(2,1)}$ and $\bar{F}^{(2,2)}$ are supported in $B_R$ and the following holds
\begin{align}
\l \bar{F}^{(2,1)} \r_{H^{N-1}} + \l \bar{F}^{(2,2)} \r_{H^{N-1}} & \lesssim \e , \label{estim Fbar 2}
\end{align}
with a constant depending only on $\delta$ and $R$.
\end{lem}

\begin{proof}
The support property follows from the support properties of $F_0$, $\bar{F}^{(1)}$ and $\omega^{(2)}$. The estimate \eqref{estim Fbar 2} follows from the definitions \eqref{def F21} and \eqref{def F22} and the estimates \eqref{estim g0 K0}, \eqref{estim F0 chap 2} and \eqref{estim Fbar 1} (recall also \eqref{def omega2} for the definition of $\omega^{(2)}$).
\end{proof}

\subsubsection{The $\lambda^1$ Hamiltonian level}\label{section ham 1}

We now turn to the $\lambda^1$ Hamiltonian level, which is solved thanks to $\ffi^{(3)}$. More precisely, since $\tau=\GO{\lambda}$ we need $\ffi^{(3)}$ to satisfy
\begin{align}
8 |\nabla u_0|_{\bar{g}_0}^2\dr_\theta^2\ffi^{(3)} & = - 8 \mathbf{H}^{[-1]}\left(\ffi^{(2)}\right)  + R^{(1)} \label{eq ffi 3}
\\&\quad - \left| K^{(0)} \cdot \left( K^{(1)} - \frac{1}{3}\tau^{(1)} \bar{g}_0 \right) \right|_{\bar{g}_0} + \cos\left(\frac{u_0}{\lambda}\right) \left| K^{(0)} \right|^2_{\bar{F}^{(1)}},\nonumber
\end{align}
where $R^{(1)}$ is defined in Lemma \ref{lem R(gamma)} and $\mathbf{H}^{[-1]}\left(\ffi^{(2)}\right) $ in Lemma \ref{lem expansion laplace}. In the RHS of \eqref{eq ffi 3}, it remains to define $K^{(1)}$ and $\tau^{(1)}$. Following the discussion of Section \ref{section strategy TT tensor}, we define
\begin{align}
K^{(1)}_{ij} & = -\half \cos \left( \frac{u_0}{\lambda}\right) \bigg(  -N_0\bar{F}^{(1)}_{ij} + \left( -\bar{g}_0^{k\ell}(K_0)_{(i\ell} + N_0^\ell \Gamma(g_0)^k_{\ell (i} \right) \bar{F}^{(1)}_{j)k}  \label{def K1}
\\&\hspace{3cm} +  \frac{1}{|\nabla u_0|_{\bar{g}_0}} \left(- (K_0)^{k\ell}   \dr_k u_0 \dr_\ell u_0 -   \bar{g}_0^{k\ell} \dr_k u_0  \dr_\ell |\nabla u_0|_{\bar{g}_0}  +\half \bar{g}_0^{k\ell}\dr_k\dr_\ell u_0 \right) \bar{F}^{(1)}_{ij}  \bigg) \nonumber
\\&\quad -\half   \left|\nabla u_0 \right|_{\bar{g}_0} \left( \cos\left( \frac{u_0}{\lambda}\right) \bar{F}^{(2,1)}_{ij} - 2\sin\left( \frac{2u_0}{\lambda}\right)\bar{F}^{(2,2)}_{ij}   \right) . \nonumber
\end{align}
Since $\tau$ is the trace with respect to $\gamma$ of $K$ we set
\begin{align}
\tau^{(1)} &  = -\cos\left(\frac{u_0}{\lambda}\right) \left| \bar{F}^{(1)} \cdot K_0 \right|_{\bar{g}_0} + \tr_{\bar{g}_0} K^{(1)}, \label{def tau 1}
\end{align}
where we used the expansion of the inverse of $\gamma$ given by \eqref{inverse gamma}. This actually fully defines the parameter $\tau$, i.e
\begin{align}
\tau = \lambda \tau^{(1)}.\label{def tau}
\end{align}

\saut
As above, we need to show that the RHS of \eqref{eq ffi 3} is purely oscillating in order to find $\ffi^{(3)}$ solution of this equation. From \eqref{def K0}, \eqref{def K1} and \eqref{def tau 1} we obtain
\begin{align*}
K^{(0)} \simf 1 + \sin(\theta) \quad \text{and} \quad K^{(1)} ,  \; \tau^{(1)}  \simf  \cos(\theta) + \sin(2\theta),
\end{align*} 
which together with \eqref{freq R1}, \eqref{H-1 a} and \eqref{def ffi 2}
gives
\begin{align*}
\text{RHS of \eqref{eq ffi 3}} \simf \cos(\theta) + \sin(2\theta) + \cos(3\theta).
\end{align*}
This shows that we can solve \eqref{eq ffi 3} by setting
\begin{align}
\ffi^{(3)} & = \cos\left(\frac{u_0}{\lambda}\right) \ffi^{(3,1)} + \sin\left(\frac{2u_0}{\lambda}\right)\ffi^{(3,2)} + \cos\left(\frac{3u_0}{\lambda}\right)\ffi^{(3,3)}\label{def ffi 3},
\end{align}
for some $\ffi^{(3,i)}$ supported in $B_R$ and satisfying
\begin{align}
\sum_{i=1,2,3} \l \ffi^{(3,i)} \r_{H^{N-3}} & \lesssim \e . \label{estim ffi 3}
\end{align}
The estimate \eqref{estim ffi 3} follows from \eqref{estim R1 et R2} and \eqref{def omega2}. Since the ansatz constructed in this article is an order 2 ansatz, the term $\ffi^{(3)}$ will be ultimately hidden in the remainder $\bar{\h}_\lambda$ (see \eqref{g bar theo chap 2} in Theorem \ref{theo main chap 2}) and thus we don't need a precise expression of $\ffi^{(3,i)}$.

\subsection{The approximate momentum constraint}\label{section mom 0}

In this section we solve the first level of the momentum constraint.  Since $\tau = \lambda \tau^{(1)}$ we have $\d\tau =\frac{1}{\lambda}\d u_0 \dr_\theta \tau^{(1)} + \lambda\d\tau^{(1)}$ where the derivatives in $\d\tau^{(1)}$ don't hit the oscillating parts of $\tau^{(1)}$. Therefore the first momentum level corresponds to the $\lambda^0$ level of \eqref{momentum constraint}, solved by $W^{(2)}$. 

\begin{remark}\label{remark what if}
This is where the assumption $\tr_{\bar{g}_0}K_0=0$ simplifies our construction. If $\tr_{\bar{g}_0}K_0\neq 0$ then $\tau$ needs to include a non-oscillating $\lambda^0$ term in $\tau^{(0)}$. This would require a non-oscillating term in $W$ at the $\lambda^0$ order to absorb it
\begin{align*}
W=w^{(0)} + \GO{\lambda},
\end{align*}
where we use a lowercase letter to emphasize that $w^{(0)}$ is non-oscillating. However, as explained in Section \ref{section conformal laplacian},  this term would produce a $\lambda^{-1}$ term in $\dive_\gamma L_\gamma W$ (precisely $M^{[-1]}(w^{(0)})$, see Lemma \ref{lem conf laplace non osc}) and thus would require an oscillating term $W^{(1)}$ in $W$ to absorb it. We make the assumption $\tr_{\bar{g}_0}K_0=0$ precisely to avoid these technicalities. 
\end{remark}

More precisely, $W^{(2)}$ needs to solve
\begin{align}
 \mathbf{M}^{[-2]}_\ell(W^{(2)}) & =  \frac{2}{3} \dr_\ell u_0 \dr_\theta \tau^{(1)}. \label{momentum 0}
\end{align}
where we used the expansion obtained in Lemma \ref{lem conf laplace osc}.  Thanks to Lemma \ref{inversion M-2} this equation rewrites as
\begin{align}
\dr_\theta^2 W^{(2)}_{\ell}  & = - \frac{1}{2  |\nabla u_0|_{\bar{g}_0}} (N_0)_\ell \dr_\theta \tau^{(1)} .\label{def W2}
\end{align}
Since the RHS of this equation is a $\dr_\theta$ derivative it is purely oscillating and we can integrate twice to obtain $W^{(2)}$.  Using \eqref{def tau 1} this also gives
\begin{align}
 W^{(2)} &\simf  \sin(\theta) + \cos(2\theta).   \label{W2 oscillation}
\end{align}

\saut
Note that as opposed to the Hamiltonian constraint, we don't solve the $\lambda^1$ momentum level here. Indeed, as Lemma \ref{lem conf laplace non osc} shows, the operator $\dive_\gamma L_\gamma$ loses one $\lambda$ power even when applied to a non-oscillating field like $\tW$ (the remainder in \eqref{ansatz W}). Therefore the $\lambda^1$ momentum level also involves $\tW$ and the equation for $W^{(3)}$ is coupled with the remainder in the ansatz \eqref{ansatz W}, as opposed to $\ffi^{(3)}$ which is not coupled to $\tffi$ (the remainder in \eqref{ansatz ffi}), since $\Delta_\gamma \tffi =\GO{\lambda^0}$.

\subsection{An almost TT-tensor}\label{section almost TT}

In this section we define the first terms in the expansion of the parameter $\sigma$ of the conformal method, i.e $\sigma^{(0)}$ and $\sigma^{(1)}$.  As explained in Section \ref{section strategy TT tensor}, the definition of the first orders of $\sigma$ follows simply from our constraint
\begin{align}
\left(\sigma + L_\gamma W\right)^{(\leq 1)} = \left( K - \frac{1}{3}\tau \gamma \right)^{(\leq 1)}.\label{condition sigma 0 1}
\end{align}
Since $W$ is given by $W=\lambda^2 \left( W^{(2)}+ \tW\right) +\GO{\lambda^3}$ (with $W^{(2)}$ defined by \eqref{def W2} and $\tW$ non-oscillating),  we use Lemma \ref{lem conformal killing} to compute $(L_\gamma W)^{(\leq 1)}$ and \eqref{condition sigma 0 1} forces us to define

\begin{align}
\sigma^{(0)}  & = K^{(0)},  \label{sigma 0}
\\ \sigma^{(1)}  & = K^{(1)} - \frac{1}{3}\tau^{(1)}\bar{g}_0  - \mathbf{K}^{[-1]}(W^{(2)}), \label{sigma 1}
\end{align}
where $K^{(0)}$, $K^{(1)}$ and $\tau^{(1)}$ are given by \eqref{def K0}, \eqref{def K1} and \eqref{def tau 1} respectively and $\mathbf{K}^{[-1]}(W^{(2)})$ is defined in Lemma \ref{lem conformal killing}.

\saut
In this section, we prove that $\sigma^{(0)} + \lambda \sigma^{(1)}  $ is \textit{almost} a TT-tensor, that is 
\begin{align}
\tr_\gamma \left(\sigma^{(0)} + \lambda \sigma^{(1)}  \right) = \GO{\lambda^2} \quad \text{and} \quad \dive_\gamma \left(\sigma^{(0)} + \lambda \sigma^{(1)}  \right) = \GO{\lambda}.\label{almost TT}
\end{align}
Using the expansion of $\gamma^{ij}$ given by \eqref{inverse gamma} we have
\begin{align*}
\tr_\gamma \left( \sigma^{(0)} + \lambda \sigma^{(1)} \right) & = \tr_{\bar{g}_0}\sigma^{(0)} + \lambda \left( -\cos\left(\frac{u_0}{\lambda}\right)\tr_{\bar{F}^{(1)}}\sigma^{(0)} + \tr_{\bar{g}_0}\sigma^{(1)}  \right) + \GO{\lambda^2}.
\end{align*}
Using now the expansion of the divergence operator $\dive_\gamma$ obtained in Lemma \ref{lem divergence gamma} we have
\begin{align*}
\dive_\gamma \left( \sigma^{(0)} + \lambda \sigma^{(1)} \right) & = \frac{1}{\lambda} \mathbf{d}^{[-1]}(\sigma^{(0)}) +  \mathbf{d}^{[0]}(\sigma^{(0)}) +  \mathbf{d}^{[-1]}(\sigma^{(1)}) + \GO{\lambda}. 
\end{align*}
The following lemma, which shows that $\sigma^{(0)} + \lambda \sigma^{(1)}  $ is almost a TT-tensor, is crucial since it validates a posteriori our whole approximate construction.

\begin{lem}\label{lem sigma01}
We have
\begin{align}
 \tr_{\bar{g}_0}\sigma^{(0)} & = 0,\label{trace sigma01 a}
 \\ -\cos\left(\frac{u_0}{\lambda}\right)\tr_{\bar{F}^{(1)}}\sigma^{(0)} + \tr_{\bar{g}_0}\sigma^{(1)} & = 0,\label{trace sigma01 b}
\end{align}
and
\begin{align}
\mathbf{d}^{[-1]}(\sigma^{(0)})  &= 0,\label{dive sigma01 a}
\\  \mathbf{d}^{[0]}(\sigma^{(0)}) +  \mathbf{d}^{[-1]}(\sigma^{(1)})  & = 0.\label{dive sigma01 b}
\end{align}
\end{lem}

\begin{proof}
We start with the trace identities. Since $\sigma^{(0)}=K^{(0)}$, \eqref{trace sigma01 a} follows from \eqref{maximal} and $\tr_{\bar{g}_0}\bar{F}^{(1)}=0$. Moreover, from \eqref{sigma 0} and \eqref{sigma 1} we have
\begin{align*}
 -\cos\left(\frac{u_0}{\lambda}\right)\tr_{\bar{F}^{(1)}}\sigma^{(0)} + \tr_{\bar{g}_0}\sigma^{(1)} & =  -\cos\left(\frac{u_0}{\lambda}\right)\tr_{\bar{F}^{(1)}}K^{(0)} + \tr_{\bar{g}_0}K^{(1)} - \tau^{(1)} 
 \\&\quad - \tr_{\bar{g}_0}\mathbf{K}^{[-1]}(W^{(2)})
 \\& = -\cos\left(\frac{u_0}{\lambda}\right)\tr_{\bar{F}^{(1)}}K^{(0)} + \tr_{\bar{g}_0}K^{(1)} - \tau^{(1)}
 \\&= 0,
\end{align*}
where we use \eqref{trace K-1} and \eqref{def tau 1}.  This proves \eqref{trace sigma01 b}.

\saut
We now look at the divergence identities. From \eqref{sigma 0}, \eqref{def K0} and \eqref{d-1} we obtain
\begin{align*}
\mathbf{d}^{[-1]}_\ell(\sigma^{(0)}) & = \half |\nabla u_0|^2_{\bar{g}_0} \cos\left(\frac{u_0}{\lambda}\right)  \bar{F}^{(1)}_{N_0\ell} = 0,
\end{align*}
which proves \eqref{dive sigma01 a}. We now compute the two parts of \eqref{dive sigma01 b}.
Using \eqref{d0} and \eqref{sigma 0} we obtain
\begin{align*}
 \mathbf{d}^{[0]}_\ell(\sigma^{(0)}) & =  \dive_{\Bar{g}_0} K^{(0)}_\ell - \Bar{g}_0^{ij} (\Tilde{\Gamma}^{(0)})^a_{i\ell} K^{(0)}_{aj} 
 \\& = (\dive_{\bar{g}_0}K_0)_\ell + \half \sin\left(\frac{u_0}{\lambda}\right)\left( \dive_{\bar{g}_0}|\nabla u_0|_{\bar{g}_0}\bar{F}^{(1)}\right)_\ell  
 \\&\quad - \Bar{g}_0^{ij} (\Tilde{\Gamma}^{(0)})^a_{i\ell} (K_0)_{aj} -  \half\sin\left(\frac{u_0}{\lambda}\right)  |\nabla u_0|_{\bar{g}_0} \Bar{g}_0^{ij} (\Tilde{\Gamma}^{(0)})^a_{i\ell} \bar{F}^{(1)}_{aj}.
\end{align*}
We use \eqref{gamma tilde 0} and \eqref{energie condition theo chap 2} to rewrite the second line in the previous expression and obtain
\begin{align*}
 \mathbf{d}^{[0]}_\ell(\sigma^{(0)}) & = (\dive_{\bar{g}_0}K_0)_\ell + \half \sin\left(\frac{u_0}{\lambda}\right)\left( \dive_{\bar{g}_0}|\nabla u_0|_{\bar{g}_0}\bar{F}^{(1)}\right)_\ell  \nonumber
 \\&\quad + \frac{1}{2} \sin\left(\frac{u_0}{\lambda}\right) \dr_\ell u_0  \left| \bar{F}^{(1)} \cdot K_0 \right|_{\bar{g}_0} +  2 \sin^2\left(\frac{u_0}{\lambda}\right)  |\nabla u_0|_{\bar{g}_0}   \dr_\ell u_0  F_0^2. \nonumber
\end{align*}
Using the background momentum constraint \eqref{momentum background}, we see that the non-oscillating terms cancel each other and we are left with
\begin{align}
 \mathbf{d}^{[0]}_\ell(\sigma^{(0)}) & = \half \sin\left(\frac{u_0}{\lambda}\right)\left( \left( \dive_{\bar{g}_0}|\nabla u_0|_{\bar{g}_0}\bar{F}^{(1)}\right)_\ell  + \dr_\ell u_0 \left| \bar{F}^{(1)} \cdot K_0 \right|_{\bar{g}_0}  \right) \label{inter d0 sigma 0}
\\&\quad - \cos\left(\frac{2u_0}{\lambda}\right) |\nabla u_0|_{\bar{g}_0}   \dr_\ell u_0  F_0^2 \nonumber .
\end{align}  
We now look at $ \mathbf{d}^{[-1]}_\ell(\sigma^{(1)})$. Using \eqref{d-1} and \eqref{sigma 1} we obtain
\begin{align*}
\mathbf{d}^{[-1]}_\ell(\sigma^{(1)}) & =  -|\nabla u_0|_{\bar{g}_0}\dr_\theta K^{(1)}_{N_0\ell}  - \frac{1}{3}\dr_\ell u_0 \dr_\theta \tau^{(1)} - \mathbf{M}_\ell^{[-2]}(W^{(2)})
\\& =\dr_\theta \left(  -|\nabla u_0|_{\bar{g}_0}  K^{(1)}_{N_0\ell}  -\dr_\ell u_0  \tau^{(1)}\right),
\end{align*}
where we use the equation satisfied by $W^{(2)}$ (see \eqref{momentum 0}). We use \eqref{def K1}, \eqref{def F21} and \eqref{def F22} to obtain
\begin{align*}
K^{(1)}_{N_0\ell} & = -\cos \left( \frac{u_0}{\lambda}\right)  \left(  \frac{ \bar{F}^{(1)}_{k\ell} }{2} \left(  \D_{N_0}N_0^k -    \bar{g}_0^{kj} (K_0)_{N_0j} \right) \right.
\\&\left. \qquad\qquad\qquad\qquad +  \frac{\bar{F}^{(1)}_{j k}}{4}  \left( \bar{g}_0^{ik} \dr_i N_0^j - \half  N_0 \bar{g}_0^{j k} + K_0^{jk} \right) (N_0)_\ell + \half \left|\nabla u_0 \right|_{\bar{g}_0} \omega^{(2)}_{N_0 \ell}  \right)
\\&\quad +  \frac{3}{4}  \sin\left( \frac{2u_0}{\lambda}\right) \left|\nabla u_0 \right|_{\bar{g}_0} F_0^2 (N_0)_\ell .
\end{align*}
Using \eqref{def tau 1} and $\tr_{\bar{g}_0}\omega^{(2)}=0$ we also obtain
\begin{align}
\tau^{(1)} &  = -\frac{3}{4}  \cos\left( \frac{u_0}{\lambda}\right)   \bar{F}^{(1)}_{k j}  \left( \bar{g}_0^{ij} \dr_i N_0^k - \half  N_0 \bar{g}_0^{j k} + K_0^{jk}  \right)   +  \frac{1}{4} \sin\left(\frac{2u_0}{\lambda}\right) |\nabla u_0|_{\bar{g}_0}F_0^2 .\label{expression tau 1}
\end{align}
This gives

\begin{align*}
 -|\nabla u_0|_{\bar{g}_0} &  K^{(1)}_{N_0\ell}  -\dr_\ell u_0  \tau^{(1)} 
\\& = |\nabla u_0|_{\bar{g}_0} \cos \left( \frac{u_0}{\lambda}\right)  \left(  \frac{  \bar{F}^{(1)}_{k\ell}}{2} \left(  \D_{N_0}N_0^k -    \bar{g}_0^{kj} (K_0)_{N_0j} \right) \right.
\\&\left. \qquad\qquad\qquad\qquad\qquad -  \frac{\bar{F}^{(1)}_{j k}}{2}  \left( \bar{g}_0^{ik} \dr_i N_0^j - \half  N_0 \bar{g}_0^{j k} + K_0^{jk}  \right) (N_0)_\ell + \half\left|\nabla u_0 \right|_{\bar{g}_0} \omega^{(2)}_{N_0 \ell}  \right)
\\&\quad +  \frac{1}{2}  \sin\left( \frac{2u_0}{\lambda}\right) \left|\nabla u_0 \right|_{\bar{g}_0} F_0^2 \dr_\ell u_0.
\end{align*}
Adding this to the expression of $ \mathbf{d}^{[0]}_\ell(\sigma^{(0)})$ given by \eqref{inter d0 sigma 0} we notice that the terms oscillating like $2\theta$ cancel (see Remark \ref{remarque magique} below) and we obtain
\begin{align*}
 &  \mathbf{d}^{[0]}_\ell(\sigma^{(0)}) +  \mathbf{d}^{[-1]}_\ell(\sigma^{(1)})
 \\&\quad =\sin\left(\frac{u_0}{\lambda}\right)\left[  \half  \left( \dive_{\bar{g}_0}|\nabla u_0|_{\bar{g}_0}\bar{F}^{(1)}\right)_\ell   -   \frac{|\nabla u_0|_{\bar{g}_0}}{2} \bar{F}^{(1)}_{k\ell} \left(  \D_{N_0}N_0^k -    \bar{g}_0^{kj} (K_0)_{N_0j} \right)   \right.
 \\&\quad\quad \left.  \qquad\qquad\qquad\qquad - \half \left|\nabla u_0 \right|^2_{\bar{g}_0} \omega^{(2)}_{N_0 \ell} - \half \dr_\ell u_0 \bar{F}^{(1)}_{k j}   \bar{g}_0^{ij} \dr_i N_0^k      +      \frac{1}{4} \dr_\ell u_0 \bar{F}^{(1)}_{k j}     N_0  \bar{g}_0^{k j}  \right].
\end{align*}
Using $\bar{F}^{(1)}_{N_0 i}=0$ we can compute the divergence of $|\nabla u_0|_{\bar{g}_0}\bar{F}^{(1)}$ and using in addition that $\omega^{(2)}_{N_0N_0}=0$ we obtain
\begin{align*}
\mathbf{d}^{[0]}_{N_0}(\sigma^{(0)}) +  \mathbf{d}^{[-1]}_{N_0}(\sigma^{(1)}) & = 0.
\end{align*}
The tangential components of $\mathbf{d}^{[0]}(\sigma^{(0)}) +  \mathbf{d}^{[-1]}(\sigma^{(1)})$ are given by
\begin{align*}
 &  \mathbf{d}^{[0]}_{\mathbf{j}}(\sigma^{(0)}) +  \mathbf{d}^{[-1]}_{\mathbf{j}}(\sigma^{(1)})
 \\&\quad  =\sin\left(\frac{u_0}{\lambda}\right)\left[  \half  \left( \dive_{\bar{g}_0}|\nabla u_0|_{\bar{g}_0}\bar{F}^{(1)}\right)_{\mathbf{j}}   -   \frac{|\nabla u_0|_{\bar{g}_0}}{2} \bar{F}^{(1)}_{k\mathbf{j}} \left(  \D_{N_0}N_0^k -    \bar{g}_0^{kj} (K_0)_{N_0j} \right)   - \half \left|\nabla u_0 \right|^2_{\bar{g}_0} \omega^{(2)}_{N_0 \mathbf{j}}  \right].
\end{align*}
This previous expression vanishes thanks to the choice of $\omega^{(2)}$ made in \eqref{def omega2}. This concludes the proof of \eqref{dive sigma01 b}.

\end{proof}

\begin{remark}\label{remarque magique}
The cancellation of the $\cos\left( \frac{2u_0}{\lambda}\right)$ terms coming from $ \mathbf{d}^{[0]}_\ell(\sigma^{(0)}) $ and $ \mathbf{d}^{[-1]}_\ell(\sigma^{(1)}) $ in their sum seems to be linked to the weak polarized null condition satisfied by the semi-linear terms in the Einstein equations, which involved products of derivatives of the metric (see Section 3.1.3 in \cite{Touati2022a} for the definition of this condition). Indeed, these terms correspond to terms of the form $\Gamma K$, $\Gamma\Gamma$ or $KK$ with $\Gamma$ the Christoffel symbols and $K$ the second fundamental form, i.e terms of the form $\dr g \dr g$.
\end{remark}

\subsection{An exact TT-tensor}\label{section TT}

In the next section, we are going to solve completely the constraint equations, i.e solve for the remainders in the high-frequency ansatz \eqref{ansatz ffi}-\eqref{ansatz W}. We will thus need the full expression of the parameters of the conformal method. While $\gamma$ and $\tau$ are already fully defined, $\sigma$ is only partially known yet and is only an \textit{almost} TT-tensor, as it was shown in the previous section.  In this section we finish the construction of $\sigma$. We choose the following ansatz
\begin{align}
\sigma & = \sigma^{(0)} + \lambda \sigma^{(1)} + \lambda^2 \left( \sigma^{(2)} +L_\gamma Y  \right) + \frac{\lambda^3}{3} \f\gamma.\label{ansatz sigma}
\end{align}
In this expression, $\sigma^{(0)}$ and $\sigma^{(1)}$ are given by \eqref{sigma 0} and \eqref{sigma 1} and $\sigma^{(2)}$, $Y$ and $\f$ are yet to be defined such that 
\begin{equation}
\tr_\gamma\sigma=0\quad\text{and}\quad\dive_\gamma \sigma =0.\label{TT}
\end{equation}
Let us explain the ansatz \eqref{ansatz sigma}. Since we need to satisfy the compatibility with the spacetime ansatz \eqref{condition sigma 0 1}, we can't modifiy the order $\lambda^0$ and $\lambda^1$ of $\sigma$. Thus, a non-oscillating remainder can only appear at the order $\lambda^2$. However, such a remainder would not be able to solve the $\lambda^1$ level of $\dive_\gamma \sigma =0$ (recall \eqref{almost TT}). Therefore we need to add an oscillating field at the order $\lambda^2$, i.e $\sigma^{(2)}$. This field will also be able to solve the $\lambda^2$ level of $\tr_\gamma \sigma=0$. Finally the remainder is chosen of the form $L_\gamma Y + \frac{\lambda}{3}\f \gamma$,  where the vector field $Y$ ensures $\dive_\gamma \sigma =0$ and the scalar function $\f$ ensures $\tr_\gamma \sigma =0$.

\saut
We now derive the equations for $\sigma^{(2)}$, $Y$ and $\f$, which illustrates the above discussion. Thanks to Lemma \ref{lem sigma01} the equations \eqref{TT} rewrite as
\begin{align*}
 \lambda^2\left( \tr_{\gamma^{(\geq 2)}} \sigma^{(0)} +  \tr_{\gamma^{(\geq 1)}} \sigma^{(1)}  +  \tr_{\gamma} \sigma^{(2)}\right) + \lambda^3 \f & = 0,
 \end{align*}
 and
 \begin{align*}
   \lambda\left( \mathbf{d}^{[1]}(\sigma^{(0)}) + \mathbf{d}^{[0]}(\sigma^{(1)}) + \mathbf{d}^{[-1]}(\sigma^{(2)}) + M^{[-1]}(Y)  \right)\qquad\qquad\qquad\qquad\qquad\qquad&
\\ + \lambda^2 \left(   \mathbf{d}^{[\geq 2]}(\sigma^{(0)}) + \mathbf{d}^{[\geq 1]}(\sigma^{(1)}) + \mathbf{d}^{[\geq 0]}(\sigma^{(2)}) + M^{[\geq 0]}(Y) + \frac{\lambda}{3} \d\f  \right) & =0,
\end{align*}
where $\d\f$ also includes derivatives of the oscillating parts of $\f$, which implies in particular that $\d\f=\GO{\lambda^{-1}}$. In order to solve these two equations, we want $\sigma^{(2)}$, $Y$ and $\f$ to satisfy the following coupled system (recall the expression of $\mathbf{d}^{[-1]}$ given by \eqref{d-1}):
\begin{align}
\tr_{\bar{g}_0}\sigma^{(2)} & = -\tr_{\gamma^{( 2)}} \sigma^{(0)} -  \tr_{\gamma^{( 1)}} \sigma^{(1)}, \label{trace sigma2}
\\ \f & = - \tr_{\gamma^{(\geq 1)}} \sigma^{(2)} -  \tr_{\gamma^{(\geq 2)}} \sigma^{(1)}   - \tr_{\gamma^{(\geq 3)}} \sigma^{(0)}, \label{eq f}
\\ -|\nabla u_0|_{\bar{g}_0} \dr_\theta \sigma^{(2)}_{N_0\ell} & =- M^{[-1]}_\ell(Y) - \mathbf{d}^{[1]}_\ell(\sigma^{(0)}) - \mathbf{d}^{[0]}_\ell(\sigma^{(1)}), \label{divergence sigma2}
\\ M^{[\geq 0]}(Y)  & = - \mathbf{d}^{[\geq 2]}(\sigma^{(0)}) - \mathbf{d}^{[\geq 1]}(\sigma^{(1)}) - \mathbf{d}^{[\geq 0]}(\sigma^{(2)}) -  \frac{\lambda}{3} \d\f. \label{eq Y}
\end{align}
Equations \eqref{trace sigma2} and \eqref{eq f} ensure $\tr_\gamma \sigma =0$ while \eqref{divergence sigma2} and \eqref{eq Y} ensure $\dive_\gamma \sigma =0$.  The rest of this section is devoted to the resolution of the system \eqref{trace sigma2}-\eqref{eq f}-\eqref{divergence sigma2}-\eqref{eq Y}. It presents a triangular structure, despite the term $M^{[-1]}(Y)$ in \eqref{divergence sigma2}.

\subsubsection{Definition of $\sigma^{(2)}$ and $\f$}

We start by solving the non-differential equations of the previous system, that is \eqref{trace sigma2}-\eqref{eq f}-\eqref{divergence sigma2}. The first step is to show that the RHS of \eqref{divergence sigma2} is purely oscillating, which is a necessary condition since the LHS is a $\dr_\theta$ derivative. Thanks to \eqref{m-1}, the first term in the RHS is oscillating,  and the next lemma deals with the last two.

\begin{lem}
The following oscillating behaviour holds
\begin{align*}
 \mathbf{d}^{[1]}(\sigma^{(0)}) +  \mathbf{d}^{[0]}(\sigma^{(1)} )\simf \cos(\theta) + \sin(2\theta) + \cos(3\theta).
\end{align*}
\end{lem}

\begin{proof}
From \eqref{d0 oscillation} we have \[ \mathbf{d}^{[0]}(\sigma^{(1)} )\simf (1+\sin(\theta)) \sigma^{(1)},\] and from \eqref{sigma 1} we have
\begin{align*}
\sigma^{(1)} & \simf K^{(1)} + \cos(\theta) + \tau^{(1)} + \mathbf{K}^{[-1]}(W^{(2)})
\\&\simf \cos(\theta) + \sin(2\theta) + \dr_\theta W^{(2)},
\end{align*}
where we used \eqref{def K1} and \eqref{def tau 1}. Now using \eqref{W2 oscillation} we conclude that 
\begin{align*}
\mathbf{d}^{[0]}(\sigma^{(1)} ) \simf \cos(\theta) + \sin(2\theta) .
\end{align*}
Now, from \eqref{d1 oscillation} we have \[\mathbf{d}^{[1]}(\sigma^{(0)}) \simf \sin(\theta) \dr_\theta \sigma^{(0)} + (\cos(\theta) + \sin(2\theta)) \sigma^{(0)},\] and from \eqref{sigma 0} and \eqref{def K0} we have
\begin{align*}
\sigma^{(0)} & \simf  1 + \sin(\theta).
\end{align*}
This concludes the proof of the lemma.
\end{proof}

We have shown that the RHS of \eqref{divergence sigma2} is purely oscillating, which allows us to formally integrate this equation in $\theta$ and obtain $\sigma^{(2)}$. More precisely,  thanks to $\bar{F}^{(1)}_{i N_0}=0$, \eqref{m-1} implies that $M_{N_0}^{[-1]}(Y)=0$. Therefore, \eqref{divergence sigma2} gives us $\sigma^{(2)}_{N_0N_0}$ as a function of lower order terms in the construction, i.e $\sigma^{(0)}$ and $\sigma^{(1)}$. Then, \eqref{trace sigma2} gives us $\sigma^{(2)}_{\1\1}$ and $\sigma^{(2)}_{\2\2}$ as a function of $\sigma^{(2)}_{N_0N_0}$, $\sigma^{(0)}$ and $\sigma^{(1)}$. All together, the diagonal components of $\sigma^{(2)}$ in the frame $(N_0,e_\1,e_\2)$ are functions of $\sigma^{(0)}$ and $\sigma^{(1)}$ satisfying
\begin{align}
\left| \sigma^{(2)}_{N_0N_0} \right| + \left| \sigma^{(2)}_{\1\1} \right| + \left| \sigma^{(2)}_{\2\2} \right| \lesssim \left| \mathbf{d}^{[1]}(\sigma^{(0)}) \right| + \left| \mathbf{d}^{[0]}(\sigma^{(1)}) \right| + \left| (\gamma^{-1})^{(2)} \sigma^{(0)}  \right| + \left| (\gamma^{-1})^{(1)} \sigma^{(1)}   \right|, \label{estim sigma 2 diago}
\end{align}
with a high-frequency behaviour, meaning that we lose one power of $\lambda$ for each derivatives.

\begin{remark}
Note that we don't impose conditions on $\sigma^{(2)}_{\1\1}$ and $\sigma^{(2)}_{\2\2}$ separately but only on their sum.
\end{remark}

The other components of $\sigma^{(2)}$ in the frame $(N_0,e_\1,e_\2)$ as well as the scalar function $\f$ depends on the vector field $Y$, which is yet to be defined. The equation \eqref{divergence sigma2} allows us to define $\sigma^{(2)}_{N_0\mathbf{i}}$ as a linear function of $Y$ and of $\sigma^{(0)}$ and $\sigma^{(1)}$.  More precisely we obtain
\begin{align}
\left| \sigma^{(2)}_{N_0\1}(Y)\right| + \left| \sigma^{(2)}_{N_0\2}(Y)\right| \lesssim \left|  \bar{F}^{(1)} Y  \right|  +  \left| \mathbf{d}^{[1]}(\sigma^{(0)}) \right| + \left| \mathbf{d}^{[0]}(\sigma^{(1)}) \right| , \label{estim sigma 2 non diago}
\end{align}
with a high-frequency behaviour. Since this component doesn't appear in the equations $\sigma^{(2)}$ needs to solve, we set $\sigma^{(2)}_{\1\2}=0$.

\saut
The scalar function $\f$ is actually already defined by \eqref{eq f}, but as $\sigma^{(2)}$ is a function of $Y$, so is $\f$. Therefore, thanks to \eqref{estim sigma 2 diago} and \eqref{estim sigma 2 non diago} $\f$ satisfies
\begin{align}
| \f(Y) | & \lesssim  \left| (\gamma^{-1})^{(\geq 1)} \sigma^{(2)}  \right|  +   \left| (\gamma^{-1})^{(\geq 2)} \sigma^{(1)}   \right|  +  \left| (\gamma^{-1})^{(\geq 3)} \sigma^{(0)}   \right|\nonumber
\\&\lesssim \left| (\gamma^{-1})^{(\geq 1)} \bar{F}^{(1)} Y  \right|  +  \left| (\gamma^{-1})^{(\geq 1)}\mathbf{d}^{[1]}(\sigma^{(0)}) \right| + \left|(\gamma^{-1})^{(\geq 1)} \mathbf{d}^{[0]}(\sigma^{(1)}) \right| \nonumber
\\&\quad + \left| (\gamma^{-1})^{(\geq 1)} (\gamma^{-1})^{(2)} \sigma^{(0)}  \right| + \left| (\gamma^{-1})^{(\geq 1)} (\gamma^{-1})^{(1)} \sigma^{(1)}   \right|\nonumber
\\&\quad +   \left| (\gamma^{-1})^{(\geq 2)} \sigma^{(1)}   \right|  +  \left| (\gamma^{-1})^{(\geq 3)} \sigma^{(0)}   \right|\label{estim f} ,
\end{align}
with a high-frequency behaviour.

\subsubsection{Solving for $Y$}\label{section solving for Y}

To conclude the construction of $\sigma$, it remains to solve \eqref{eq Y}. As explained after Lemma \ref{lem conf laplace non osc}, this is done by actually "replacing" the operator $M^{[\geq 0]}$ by $\dive_e L_e$ together with a fixed point argument. More precisely, we define a map $\mathbf{\Psi}$ 
\begin{align*}
\mathbf{\Psi} : \mathcal{B} \longrightarrow \mathcal{B}
\end{align*}
such that $\mathbf{\Psi}(Y)$ is the solution of 
\begin{align}
\dive_e L_e \mathbf{\Psi}(Y)  & = \dive_e L_e Y  - M^{[\geq 0]}(Y) - \mathbf{d}^{[\geq 2]}(\sigma^{(0)})  \label{eq psi(Y)}
\\&\quad - \mathbf{d}^{[\geq 1]}(\sigma^{(1)}) - \mathbf{d}^{[\geq 0]}(\sigma^{(2)}(Y)) -  \frac{\lambda}{3} \d\f(Y),\nonumber
\end{align}
and where 
\begin{align}
\mathcal{B}=\enstq{ Z \in H^2_\delta }{ \l Z \r_{H^2_\delta}\leq C_1 \e }.\label{la boule}
\end{align}
with $C_1>0$ to be chosen later.  Note that any fixed point of $\mathbf{\Psi}$ is a solution of \eqref{eq Y}. In order to prove the existence of a fixed point, we need to show that $\mathbf{\Psi}$ is well-defined and is a contraction.

\begin{prop}\label{prop Psi}
If $C_1$ is large enough and $\e$ is small enough, then the map $\mathbf{\Psi}$ is well-defined and is a contraction.
\end{prop}

\begin{proof}
Let $Y\in \mathcal{B}$. We bound the RHS of \eqref{eq psi(Y)} in $L^2_{\delta+2}$:
\begin{align*}
\l \text{RHS of \eqref{eq psi(Y)}}  \r_{L^2_{\delta+2}} & \lesssim  \l  \dive_e L_e Y  - M^{[\geq 0]}(Y) \r_{L^2_{\delta+2}} + \l  \mathbf{d}^{[\geq 0]}(\sigma^{(2)}(Y))\r_{L^2_{\delta+2}} +\lambda \l \d[\f(Y)]\r_{L^2_{\delta+2}}
\\&\quad + \l \mathbf{d}^{[\geq 2]}(\sigma^{(0)}) + \mathbf{d}^{[\geq 1]}(\sigma^{(1)})  \r_{L^2}
\\& =\vcentcolon A+B+C+D ,
\end{align*}
where we omitted the weights for the last term since it is compactly supported. For $A$,  \eqref{estimation m0} gives
\begin{align*}
A & \lesssim \l (\gamma^{-1}-e^{-1}) \dr^2 Y \r_{L^2_{\delta+2}} + \l \dr \gamma \dr Y \r_{L^2_{\delta+2}} + \l (\dr \gamma)^2 Y  \r_{L^2_{\delta+2}} + \l (\dr^2\gamma)^{(\geq 0)} Y  \r_{L^2_{\delta+2}} ,
\end{align*}
where we also used the fact that the coefficients of $\gamma$ are bounded. We bound all the metric terms in $L^\infty$ using the background regularity \eqref{estim g0 K0} and \eqref{estim Fbar 1}. More precisely,  we have
\begin{align}
\l \gamma^{-1}-e^{-1} \r_{L^\infty} & \lesssim \l \bar{g}_0 - e \r_{L^\infty} + \lambda \l \bar{F}^{(1)} \r_{L^\infty} + \lambda^2 \l \omega^{(2)} \r_{L^\infty} ,\label{estim gamma - e}
\\  \l \dr \gamma \r_{L^\infty} & \lesssim      \l \dr \bar{g}_0 \r_{L^\infty} + \l \bar{F}^{(1)} \r_{L^\infty} \label{estim d gamma}
\\&\quad + \lambda\left(  \l \dr \bar{F}^{(1)} \r_{L^\infty} + \l \omega^{(2)} \r_{L^\infty}  \right) + \lambda^2 \l \dr\omega^{(2)} \r_{L^\infty}    , \nonumber
\\  \l (\dr^2\gamma)^{(\geq 0)} \r_{L^\infty} & \lesssim   \l \dr^2\bar{g}_0 \r_{L^\infty} + \l \dr \bar{F}^{(1)} \r_{L^\infty} + \l \omega^{(2)} \r_{L^\infty}\label{estim d2 gamma}
\\&\quad + \lambda \left( \l \dr^2\bar{F}^{(1)} \r_{L^\infty} + \l \dr\omega^{(2)} \r_{L^\infty} \right) + \lambda^2 \l \dr^2 \omega^{(2)}  \r_{L^\infty}      . \nonumber
\end{align}
Using \eqref{estim g0 K0}, \eqref{estim Fbar 1}, \eqref{def omega2} and $\l Y \r_{H^2_\delta}\leq C_1\e$, the estimates \eqref{estim gamma - e}-\eqref{estim d gamma}-\eqref{estim d2 gamma} then imply that $A\lesssim C_1\e^2$. 

The term $D$ depends only on previous terms of the construction so \eqref{d1} simply gives $D\lesssim \e$. The maps $Y\longmapsto\sigma^{(2)}(Y)$ and $Y\longmapsto\f(Y)$ are affine with coefficients as $D$ (see \eqref{estim sigma 2 diago}, \eqref{estim sigma 2 non diago} and \eqref{estim f}) so a combination of the two previous arguments gives $B+C\lesssim C_1\e^2 + \e$.  Note that for $C$ we need to compensate the loss in power of $\lambda$ when differenciating the oscillating parts of $\f$ with the $\lambda$ in front.

\saut
We have proved that $A+B+C+D \lesssim C_1\e^2 + \e$. In particular, this allows us to use the second part of Proposition \ref{prop Delta et dive e L e} to prove that there exists a unique $\mathbf{\Psi}(Y)\in H^2_\delta$ solving \eqref{eq psi(Y)}. Moreover we have
\begin{align*}
\l \mathbf{\Psi}(Y) \r_{H^2_\delta} \lesssim C_1\e^2 + \e.
\end{align*}
Therefore, taking $C_1$ large compared to the numerical constant appearing in these estimates and $\e$ small compared to 1 proves that $\mathbf{\Psi}$ is well-defined and maps $\mathcal{B}$ to itself. 

\saut
To prove that $\mathbf{\Psi}$ is a contraction, we consider $Y_a$ and $Y_b$ two elements of $\mathcal{B}$. By substracting the equations satisfied by $\mathbf{\Psi}(Y_a)$ and $\mathbf{\Psi}(Y_b)$ we obtain the equation for their difference
\begin{align}
\dive_e L_e \left( \mathbf{\Psi}(Y_a) - \mathbf{\Psi}(Y_b) \right) & = \left( \dive_e L_e - M^{[\geq 0]} \right)(Y_a - Y_b) \label{eq diff Ya-Yb}
\\&\quad - \mathbf{d}^{[\geq 0]}\left(\sigma^{(2)}(Y_a) - \sigma^{(2)}(Y_b) \right) - \frac{\lambda}{3} \d \left( \f (Y_a)-\f(Y_b) \right) \nonumber .
\end{align}
Using again the fact that $Y\longmapsto\sigma^{(2)}(Y)$ and $Y\longmapsto\f(Y)$ are affine and using \eqref{estimation m0} for $\dive_e L_e - M^{[\geq 0]}$, we can prove that
\begin{align*}
\l \text{RHS of \eqref{eq diff Ya-Yb}} \r_{L^2_{\delta+2}} & \lesssim \e \l Y_a - Y_b \r_{H^2_{\delta}} .
\end{align*}
Therefore, taking $\e$ small enough ensures that $\mathbf{\Psi}$ is a contraction.
\end{proof}

Thanks to this proposition, the Banach fixed point theorem implies the existence of $Y\in \mathcal{B}$ solving \eqref{eq Y}, and therefore $(\sigma^{(2)}(Y),\f(Y),Y)$ solves \eqref{trace sigma2}-\eqref{eq f}-\eqref{divergence sigma2}-\eqref{eq Y}. 

\saut
We can also prove that $Y$ enjoys higher regularity. Indeed we can bound the RHS of \eqref{eq psi(Y)} in higher order Sobolev spaces and use elliptic estimates for $\dive_eL_e$ as in the previous proposition. The worse term in \eqref{eq psi(Y)} is given by $\nabla \sigma^{(1)}$ which is bounded in $H^{N-2}$ (see \eqref{sigma 1}, \eqref{def K1} and \eqref{estim Fbar 2}) and in terms of decay the worse term is $\nabla \sigma^{(0)}$ (see \eqref{sigma 0}, \eqref{def K0} and \eqref{estim g0 K0}). Therefore we obtain a solution $Y$ of \eqref{eq Y} such that $Y\in H^N_\delta$ and
\begin{align}
\l  Y \r_{H^{k+2}_{\delta}} \lesssim \frac{\e}{\lambda^{k}}, \label{estim Y}
\end{align} 
for $k\in\llbracket 0, N-2\rrbracket$, where the loss of $\lambda$ powers is due to the high-frequency character of each term in \eqref{eq psi(Y)}. We summarize what we know on the parameter $\sigma$ in the following corollary.

\begin{coro}
The tensor $\sigma$ defined by \eqref{ansatz sigma} is a TT-tensor for the metric $\gamma$, belongs to $H^{N-1}_{\delta+1}$ and satisfies
\begin{align}
\max_{ k\in\llbracket 0,N-3 \rrbracket } \lambda^{k+2}  \l  \sigma  \r_{H^{k+2}_{\delta+1}} + \max_{ k\in\llbracket 0,N-3 \rrbracket } \lambda^{k}  \l \nabla^k  \sigma  \r_{L^\infty}  \lesssim \e. \label{estim sigma}
\end{align}
\end{coro}

\begin{proof}
The oscillating terms $\sigma^{(0)}$ and $\sigma^{(1)}$ lose one $\lambda$ power for each derivatives and we can estimate the actual tensors by estimating \eqref{sigma 0}-\eqref{sigma 1} directly in weighted Sobolev spaces or in $L^\infty$ using Sobolev embeddings of Proposition \ref{prop WSS chap 2}. Moreover, we can neglect $\sigma^{(2)}(Y)$ and $\f(Y)$ and focus on $L_\gamma Y$ in \eqref{ansatz sigma} which rewrites broadly as a $\nabla Y$ term since we can put the $\gamma$ and $\dr\gamma$ term in $L^\infty$. Therefore, the estimate \eqref{estim sigma} follows directly from \eqref{estim Y} and Sobolev embeddings.
\end{proof}

\section{Exact solution to the constraint equations}\label{section exact solution}

We are now ready to solve the constraint equations \eqref{hamiltonian constraint} and \eqref{momentum constraint}. The parameters of these equations are $\gamma$, $\tau$ and $\sigma$. The metric $\gamma$ and the scalar function $\tau$ are fully known thanks to Section \ref{section def gamma} and \eqref{def tau 1}. The TT-tensor $\sigma$ has been defined in Sections \ref{section TT} and \ref{section almost TT}. Recall that the solutions of \eqref{hamiltonian constraint}-\eqref{momentum constraint} are of the form
\begin{align*}
W & =  \lambda^2\left( W^{(2)}\left( \frac{u_0}{\lambda}\right) + \tW \right)+ \lambda^3 W^{(3)}\left( \frac{u_0}{\lambda}\right) ,
\\ \ffi & = 1 + \lambda^2 \left( \ffi^{(2)}\left( \frac{u_0}{\lambda}\right) + \tffi \right) + \lambda^3 \ffi^{(3)}\left( \frac{u_0}{\lambda}\right),
\end{align*}
where $\ffi^{(2)}$, $\ffi^{(3)}$ and $W^{(2)}$ are defined in Sections \ref{section ham 0}, \ref{section ham 1} and \ref{section mom 0} respectively. Therefore, it remains to construct $\tffi$, $\tW$ and $W^{(3)}$.

\subsection{System for the remainders}

The construction of Sections \ref{section ham 0}, \ref{section ham 1} and \ref{section mom 0} ensures that the constraint equations \eqref{hamiltonian constraint} and \eqref{momentum constraint} are partly solved. More precisely, it remains to solve the $\lambda^{\geq 2}$ levels of \eqref{hamiltonian constraint} and the $\lambda^{\geq 1}$ levels of \eqref{momentum constraint}. In this section, we compute the exact equations this gives for $\tffi$, $\tW$ and $W^{(3)}$.

\subsubsection{Definition of $W^{(3)}(\tW)$}

The purpose of the oscillating vector field $W^{(3)}$ is to solve the $\lambda^1$ momentum level. However, since the conformal Laplacian $\dive_\gamma L_\gamma$ loses one power of $\lambda$ even when applied to a non-oscillating field such as $\tW$ (see Lemma \ref{lem conf laplace non osc}), the latter is a source term in the equation for $W^{(3)}$. This explains why $W^{(3)}$ is considered as part of the remainders, when $\ffi^{(3)}$ is not.

\saut
We define $W^{(3)}$ as a function of $\tW$. The $\lambda^1$ momentum level writes
\begin{align}
\mathbf{M}_\ell^{[-2]}(W^{(3)}) + M^{[-1]}_\ell(\tW) + \mathbf{M}_\ell^{[-1]}(W^{(2)}) = \frac{2}{3}\dr_\ell \tau^{(1)}.
\end{align}
Thanks to \eqref{M-2 expression} and \eqref{m-1} this is equivalent to
\begin{align*}
 \dr^2_\theta W^{(3)}_\ell + \frac{1}{3}(N_0)_\ell \dr^2_\theta W^{(3)}_{N_0}    = -    \cos\left(\frac{u_0}{\lambda}\right)  \bar{g}_0^{ij}   \bar{F}^{(1)}_{i \ell}   \tW_j  + \frac{1}{|\nabla u_0|^2_{\bar{g}_0}}\left(  \frac{2}{3}\dr_\ell \tau^{(1)} -  \mathbf{M}_\ell^{[-1]}(W^{(2)}) \right) .
\end{align*}
Lemma \ref{inversion M-2} then gives
\begin{align}
 \dr^2_\theta W^{(3)}_\ell & = -    \cos\left(\frac{u_0}{\lambda}\right)  \bar{g}_0^{ij}   \bar{F}^{(1)}_{i \ell}   \tW_j  \label{def W3}
 \\&\quad+ \frac{1}{|\nabla u_0|^2_{\bar{g}_0}}\left[  \frac{2}{3}\dr_\ell \tau^{(1)} -  \mathbf{M}_\ell^{[-1]}(W^{(2)})   -  \frac{1}{4}(N_0)_\ell   \left(  \frac{2}{3}N_0 \tau^{(1)} -  \mathbf{M}_{N_0}^{[-1]}(W^{(2)}) \right) \right].\nonumber
\end{align}
Let us check that the RHS of this equation is purely oscillating. Since $\tau^{(1)}$ is purely oscillating (see \eqref{expression tau 1}), we only need to check $\mathbf{M}_\ell^{[-1]}(W^{(2)}) $. For this we use first \eqref{M-1 oscillation} and then \eqref{W2 oscillation}, this gives
\begin{align*}
\mathbf{M}^{[-1]}(W^{(2)}) & \simf   \cos(\theta) W^{(2)}   +  \cos(\theta) \dr^2_\theta W^{(2)}    + (1+\sin(\theta))   \dr_\theta W^{(2)}
\\&\simf  \cos(\theta) + \sin(2\theta) + \cos(3\theta).     
\end{align*}
The RHS of \eqref{def W3} is thus purely oscillating and we can integrate it twice with respect to $\theta$. We obtain
\begin{align}
W^{(3)}_\ell(\tW) & = \cos\left(\frac{u_0}{\lambda}\right)  \bar{g}_0^{ij}   \bar{F}^{(1)}_{i \ell}   \tW_j  +  W^{(3,rest)}_\ell , \label{def W3 bis}
\end{align}
where $W^{(3,rest)}$ satisfies
\begin{align}
\left| W^{(3,rest)} \right| & \lesssim \left|\dr\tau^{(1)} \right|  + \left| \mathbf{M}^{[-1]}(W^{(2)}) \right| , \label{estim W3 rest}
\end{align}
with a high-frequency behaviour.

\subsubsection{The system for $\tW$ and $\tffi$}

In this section, we will expand in the most concise way the non-linearities involved in the equations for $\tffi$ and $\tW$. We start with the equation for $\tW$, which, if we drop the vectorial notation, writes
\begin{align}
 M^{[\geq 0]}(\tW) & = - \mathbf{M}^{[\geq -1 ]}(W^{(3)}(\tW)) - \mathbf{M}^{[\geq 0 ]}(W^{(2)}) + \frac{2}{3}(\ffi^6 \dr\tau)^{(\geq 2)}      ,                            \label{eq tW}
\end{align}
The following lemma expands the non-linearity in \eqref{eq tW}. 

\begin{lem}\label{lem exp ffi^6}
We have 
\begin{align}
\frac{2}{3}(\ffi^6 \dr_\ell\tau)^{(\geq 2)}  & =  \mathbf{a}_0 +  \sum_{k=1}^{6}\lambda^{2(k-1)} \mathbf{a}_k \tffi^k\label{exp ffi^6} ,
\end{align}
where for $k\in \llbracket 0,6 \rrbracket$, $\mathbf{a}_k$ is supported in $B_R$ and 
\begin{align}
 \max_{i\in \llbracket  0,N-5  \rrbracket } \lambda^i \l \nabla^i  \mathbf{a}_k   \r_{L^\infty} \lesssim\e. \label{estim ak}
\end{align}
\end{lem}

\begin{proof}
Recall that $\tau=\lambda\tau^{(1)}$ implies $\dr\tau=\GO{\lambda^0}$, thus we only need to expand 
\begin{align*}
\left( 1 + \lambda^2 \left( \ffi^{(2)} + \tffi \right) + \lambda^3 \ffi^{(3)} \right)^6 ,
\end{align*}
and only keep the terms of order $\lambda^2$ or more, which only excludes the term 1. The coefficient $\mathbf{a}_0$ in \eqref{exp ffi^6} contains all the terms where $\tffi$ doesn't appear, it is thus a polynomial in $\ffi^{(2)}$ and $\ffi^{(3)}$ with no constant term and multiplied by $\dr\tau$. Therefore, $\mathbf{a}_0$ shares the same support property as $\ffi^{(2)}$ and $\ffi^{(3)}$ and the estimate \eqref{estim ak} follows from \eqref{def ffi 2}, \eqref{estim ffi 3} and \eqref{def tau 1}. If $k\in \llbracket 0,6 \rrbracket$, the same reasoning applies but $\mathbf{a}_k$ is now a polynomial in $\ffi^{(2)}$ and $\ffi^{(3)}$ with a constant term. But as this polynomial is still multiplied by $\dr\tau$, the support property and the estimate still hold.
\end{proof}

The equation for $\tffi$ writes
\begin{align}
8\Delta_\gamma \tffi & =   -  8 \sum_{i=2}^3 \mathbf{H}^{[\geq 2-i]}(\ffi^{(i)}) + R^{(\geq 2)} + R(\gamma)\left( \ffi^{(2)} + \tffi + \lambda\ffi^{(3)}\right)   \label{eq tffi} 
\\&\quad + \frac{2}{3}(\tau^2 \ffi^5)^{(\geq 2)} - \left(\left|\sigma+L_\gamma W\right|^2_{\gamma}\ffi^{-7} \right)^{(\geq 2)}.   \nonumber
\end{align}
The next two lemmas expand the non-linearities in \eqref{eq tffi}.

\begin{lem}\label{lem exp tau^2}
We have
\begin{align*}
\frac{2}{3}(\tau^2 \ffi^5)^{(\geq 2)} =      \mathbf{b}_0 +  \sum_{k=1}^{5}\lambda^{2k} \mathbf{b}_k \tffi^k       ,
\end{align*}
where for $k\in \llbracket 0,5 \rrbracket$, $\mathbf{b}_k$ is supported in $B_R$ and 
\begin{align}
 \max_{i\in \llbracket  0,N-5  \rrbracket } \lambda^i \l \nabla^i  \mathbf{b}_k   \r_{L^\infty} \lesssim\e. \label{estim bk}
\end{align}
\end{lem}

The proof of Lemma \ref{lem exp tau^2} is left to the reader since it is very similar to the one of Lemma \ref{lem exp ffi^6}. We now expand the non-linearities with a negative power of $\ffi$.

\begin{lem}\label{lem exp ffi-7}
There exists a universal constant $C_{emb}>0$ such that if $\l \tffi \r_{H^2_\sigma} < C_{emb}^{-1}$ and if $\e$ is small enough, then we have
\begin{align*}
  \left( \left|\sigma+L_\gamma W\right|^2_{\gamma} \ffi^{-7} \right)^{(\geq 2)} & = \left(\left|\sigma+L_\gamma W\right|^2_{\gamma} \right)^{(\geq 2)} + \left|\sigma+L_\gamma W\right|^2_{\gamma}\left( \mathbf{c}_0 +  \sum_{k\geq 1} \mathbf{c}_k \lambda^{2(k-1)} \tffi^k  \right)   ,
\end{align*}
where the $\mathbf{c}_k$ satisfy:
\begin{itemize}
\item $\mathbf{c}_0$ is supported in $B_R$ and we have
\begin{align}
 \max_{i\in \llbracket  0,N-5  \rrbracket } \lambda^i \l \nabla^i  \mathbf{c}_0   \r_{L^\infty} \lesssim \e, \label{estim c0}
\end{align}
\item if $k\geq 1$, we have
\begin{align}
 \max_{i\in \llbracket  0,N-5  \rrbracket } \lambda^i \l \nabla^i  \mathbf{c}_k   \r_{L^\infty} \lesssim 1. \label{estim ck geq 1}
\end{align}
\end{itemize}
\end{lem}

\begin{proof}
The constant $C_{emb}$ is the one appearing in the embedding $H^2_\delta\xhookrightarrow{} L^\infty$ (see Proposition \ref{prop WSS chap 2}), i.e
\begin{align*}
\l u \r_{L^\infty} \leq C_{emb} \l u \r_{H^2_\delta},
\end{align*} 
for all $u\in H^2_\delta$. Now if $\e$ is small enough and if $\l \tffi \r_{H^2_\delta} < C_{emb}^{-1}$, we have
\begin{align*}
\l \ffi^{(2)} + \tffi + \lambda \ffi^{(3)} \r_{L^\infty} \leq 1.
\end{align*}
This allows us to expand $\ffi^{-7}= \left( 1 + \lambda^2\left(\ffi^{(2)} + \tffi + \lambda \ffi^{(3)} \right) \right)^{-7}$.  Indeed there exists a sequence $(\mathbf{c}_k)_{k\in\N}$ such that
\begin{align*}
\ffi^{-7} = 1  + \lambda^2 \left( \mathbf{c}_0 +  \sum_{k\geq 1} \mathbf{c}_k \lambda^{2(k-1)} \tffi^k  \right)
\end{align*}
where $\mathbf{c}_0$ is a polynomial in $\ffi^{(2)}$ and $\ffi^{(3)}$ with no constant term and $\mathbf{c}_k$ for $k\geq 1$ is a polynomial in $\ffi^{(2)}$ and $\ffi^{(3)}$ with a constant term bounded but not compactly supported. This justifies the estimates \eqref{estim c0} and \eqref{estim ck geq 1}. Therefore, we have
\begin{align*}
\left( \left|\sigma+L_\gamma W\right|^2_{\gamma} \ffi^{-7} \right)^{(\geq 2)} & = \left(\left|\sigma+L_\gamma W\right|^2_{\gamma} \right)^{(\geq 2)} + \left|\sigma+L_\gamma W\right|^2_{\gamma}\left( \mathbf{c}_0 +  \sum_{k\geq 1} \mathbf{c}_k \lambda^{2(k-1)} \tffi^k  \right),
\end{align*}
which concludes the proof.
\end{proof}

Putting Lemmas \ref{lem exp ffi^6}, \ref{lem exp tau^2} and \ref{lem exp ffi-7} together, we obtain the final form of the system solved by $\tffi$ and $\tW$:
\begin{align}
M^{[\geq 0]}(\tW) & = - \mathbf{M}^{[\geq -1 ]}(W^{(3)}(\tW)) +  \sum_{k=1}^{6}\lambda^{2(k-1)} \mathbf{a}_k \tffi^k + \mathcal{R}_{\mathrm{mom}},\label{eq tW final}
\\ 8\Delta_\gamma \tffi & =     R(\gamma) \tffi   +  \sum_{k=1}^{5}\lambda^{2k} \mathbf{b}_k \tffi^k  \label{eq tffi final}
\\&\quad - \left(\left|\sigma+L_\gamma W\right|^2_{\gamma} \right)^{(\geq 2)} - \left|\sigma+L_\gamma W\right|^2_{\gamma}\left( \mathbf{c}_0 +  \sum_{k\geq 1} \mathbf{c}_k \lambda^{2(k-1)} \tffi^k  \right)    +   \mathcal{R}_{\mathrm{Ham}}    ,    \nonumber
\end{align}
where we define the following remainders
\begin{align}
\mathcal{R}_{\mathrm{mom}} & = - \mathbf{M}^{[\geq 0 ]}(W^{(2)}) +  \mathbf{a}_0 ,\label{def Rmom}
\\ \mathcal{R}_{\mathrm{Ham}} & = -  8 \sum_{i=2}^3 \mathbf{H}^{[\geq 2-i]}(\ffi^{(i)}) + R^{(\geq 2)} + R(\gamma)\left( \ffi^{(2)}  + \lambda\ffi^{(3)}\right) + \mathbf{b}_0.\label{def Rham}
\end{align}

\subsection{Fixed point argument}

In this section, we solve \eqref{eq tW final} and \eqref{eq tffi final} by a fixed point argument. As in Section \ref{section solving for Y}, the idea is to replace the operators depending on $\gamma$ by their Euclidean equivalent and use the smallness of $\gamma-e$. We introduce the map $\mathbf{\Phi} $
\begin{align*}
\mathbf{\Phi} & : \mathcal{B}\times \mathcal{B}  \longrightarrow \mathcal{B}\times \mathcal{B} 
\\ &\quad (\tffi,\tW)  \longmapsto \left(\mathbf{\Phi}_1(\tffi) , \mathbf{\Phi}_2(\tW) \right),
\end{align*}
such that $\mathbf{\Phi}_1(\tffi)$ and $\mathbf{\Phi}_2(\tW)$ are solutions of the coupled system
\begin{align}
\dive_eL_e \mathbf{\Phi}_2(\tW) & = \dive_eL_e(\tW) - M^{[\geq 0]}(\tW)  \label{eq Phi tW}
\\&\quad - \mathbf{M}^{[\geq -1 ]}(W^{(3)}(\tW)) +  \sum_{k=1}^{6}\lambda^{2(k-1)} \mathbf{a}_k \tffi^k + \mathcal{R}_{\mathrm{mom}}\nonumber ,
\\ 8 \Delta \mathbf{\Phi}_1(\tffi) & = 8 \Delta \tffi - 8\Delta_\gamma\tffi  + R(\gamma) \tffi   +  \sum_{k=1}^{5}\lambda^{2k} \mathbf{b}_k \tffi^k   \label{eq Phi tffi}
\\&\quad - \left(\left|\sigma+L_\gamma W\right|^2_{\gamma} \right)^{(\geq 2)} - \left|\sigma+L_\gamma W\right|^2_{\gamma}\left( \mathbf{c}_0 +  \sum_{k\geq 1} \mathbf{c}_k \lambda^{2(k-1)} \tffi^k  \right)    +   \mathcal{R}_{\mathrm{Ham}}    ,    \nonumber
\end{align}
and where $\mathcal{B}$ is defined in \eqref{la boule}. Note that a fixed point of $\mathbf{\Phi}$ solves \eqref{eq tW final} and \eqref{eq tffi final}. In order to apply the Banach fixed point theorem and prove the existence of a fixed point, we prove in the next proposition that $\mathbf{\Phi}$ is well-defined and is a contraction.

\begin{prop}
If $C_1$ is large enough and $\e$ is small enough, then $\mathbf{\Phi}$ is well-defined and is a contraction.
\end{prop}

\begin{proof}
Let $(\tffi,\tW)\in \mathcal{B}\times \mathcal{B}$. We start by estimating the $L^2_{\delta+2}$ norm of the RHS of \eqref{eq Phi tW}:
\begin{align*}
\l \text{RHS of \eqref{eq Phi tW}} \r_{L^2_{\delta+2}} & \lesssim \l  \dive_eL_e(\tW) - M^{[\geq 0]}(\tW) \r_{L^2_{\delta+2}} + \l \mathbf{M}^{[\geq -1 ]}(W^{(3)}(\tW)) \r_{L^2_{\delta+2}} 
\\&\quad + \l \sum_{k=1}^{6}\lambda^{2(k-1)} \mathbf{a}_k \tffi^k \r_{L^2} + \l\mathcal{R}_{\mathrm{mom}} \r_{L^2}
\\& =\vcentcolon A + B + C + D,
\end{align*}
where we omitted the weights for the last two terms since they are compactly supported. As in the proof of Proposition \ref{prop Psi}, we obtain $A\lesssim C_1\e^2$.  For $B$, we note that the operator $\mathbf{M}^{[\geq -1 ]}$ is linear and has bounded coefficients and involves up to two derivatives of $W^{(3)}(\tW)$, recall Lemma \ref{lem conf laplace osc}. Moreover, thanks to \eqref{def W3 bis} $W^{(3)}(\tW)$ is compactly supported so we obtain $B\lesssim \l W^{(3)}(\tW) \r_{H^2} $. Using \eqref{estim Fbar 1} and \eqref{estim W3 rest} this implies $B\lesssim C_1\e^2 + \e$.

For $C$, we simply estimate $\tffi$ in $L^\infty$ using the embedding $H^2_\delta\xhookrightarrow{}L^\infty$ (see Proposition \ref{prop WSS chap 2}) and together with \eqref{estim ak} this gives $C\lesssim C(C_1) \e^2$,  where $C(C_1)$ denotes a numerical constant depending on $C_1$. Using \eqref{def Rmom}, \eqref{def W2} and \eqref{estim ak} again we also obtain $D\lesssim \e$. This discussion proves that
\begin{align}
\l \text{RHS of \eqref{eq Phi tW}} \r_{L^2_{\delta+2}} & \lesssim C(C_1)\e^2 + \e.\label{estim RHS eq tW}
\end{align}

We now estimate the RHS of \eqref{eq Phi tffi}:
\begin{align*}
&\l \text{RHS of \eqref{eq Phi tffi}} \r_{L^2_{\delta+2}} 
\\&\quad \lesssim \l \Delta \tffi - \Delta_\gamma\tffi \r_{L^2_{\delta+2}} + \l R(\gamma)\tffi \r_{L^2_{\delta+2}} + \l \sum_{k=1}^{5}\lambda^{2k} \mathbf{b}_k \tffi^k  \r_{L^2} + \l \mathcal{R}_{\mathrm{Ham}} \r_{L^2_{\delta+2}}
\\&\quad\quad + \l \left(\left|\sigma+L_\gamma W\right|^2_{\gamma} \right)^{(\geq 2)} \r_{L^2_{\delta+2}} + \l \left|\sigma+L_\gamma W\right|^2_{\gamma}\left( \mathbf{c}_0 +  \sum_{k\geq 1} \mathbf{c}_k \lambda^{2(k-1)} \tffi^k  \right) \r_{L^2_{\delta+2}} 
\\& \quad =\vcentcolon A + B + C + D + E + F,
\end{align*}
where we omitted the weights for the third and sixth terms since they are compactly supported. For $A$ we use the expansion defining $\gamma$, similarly as in \eqref{estimation m0}:
\begin{align*}
A & \lesssim \l (\gamma^{-1}-e^{-1})\dr^2\tffi \r_{L^2_{\delta+2}}  +   \l \dr \gamma \dr \tffi \r_{L^2_{\delta+2}} \lesssim C_1\e^2,
\end{align*}
where we bound the metric coefficients and their derivatives in $L^\infty$ using \eqref{estim gamma - e} and \eqref{estim d gamma}. The terms $B$ and $C$ only contains $\tffi$ with zero derivatives, which we simply bound in $L^\infty$ using $H^2_\delta\xhookrightarrow{}L^\infty$. We then use Lemma \ref{lem R(gamma)} and \eqref{estim bk} to obtain $B+C\lesssim C(C_1)\e^2$. Similar arguments lead to $D\lesssim \e$.

We now estimate $E$ and $F$. It involves the TT-tensor $\sigma$ but thanks to the estimate \eqref{estim sigma} we can put it in $L^\infty$ and thus focus on $L_\gamma W$. For the same reason, we neglect $W^{(2)}$ and $W^{(3)}(\tW)$.  Since the term
\begin{align*}
\mathbf{c}_0 +  \sum_{k\geq 1} \mathbf{c}_k \lambda^{2(k-1)} \tffi^k ,
\end{align*}
can be bounded in $L^\infty$ by $C(C_1)\e$ (using \eqref{estim c0}-\eqref{estim ck geq 1} and $H^2_\delta\xhookrightarrow{}L^\infty$ for the powers of $\tffi$), in order to estimate $E$ and $F$ it is enough to estimate $\l (L_\gamma \tW)^2 \r_{L^2_{\delta+2}} $. Since $L_\gamma \tW$ contains derivatives of $\gamma$ we can't directly use the product law $H^1_{\delta+1}\times H^1_{\delta+1} \xhookrightarrow{} L^2_{\delta+2} $ of Proposition \ref{prop WSS chap 2} without losing one $\lambda$ power. Instead we expand
\begin{align*}
\l (L_\gamma \tW)^2 \r_{L^2_{\delta+2}}  & \lesssim \l  (\dr\tW)^2 \r_{L^2_{\delta+2}}   +  \l  (\dr\gamma)^2 (\tW)^2 \r_{L^2_{\delta+2}}   +  \l   \dr\gamma \tW \dr\tW  \r_{L^2_{\delta+2}} .
\end{align*}
 For the first term we use the product law $H^1_{\delta+1}\times H^1_{\delta+1} \xhookrightarrow{} L^2_{\delta+2} $ of Proposition \ref{prop WSS chap 2}. For the second and third terms, we bound $\dr\gamma$ in $L^\infty$ (recall \eqref{estim d gamma}) and use the product laws $H^2_\delta\times H^2_\delta \xhookrightarrow{} L^2_{\delta+2} $ and $H^2_{\delta}\times H^1_{\delta+1} \xhookrightarrow{} L^2_{\delta+2} $. We obtain $\l (L_\gamma \tW)^2 \r_{L^2_{\delta+2}}   \lesssim C(C_1)\e^2$ and
\begin{align*}
E+F\lesssim C(C_1)\e^2 + \e.
\end{align*}
This discussion proves that 
\begin{align}
\l \text{RHS of \eqref{eq Phi tffi}} \r_{L^2_{\delta+2}} & \lesssim C(C_1)\e^2 + \e.\label{estim RHS eq tffi}
\end{align}
Using the first part of Proposition \ref{prop Delta et dive e L e}, \eqref{estim RHS eq tW} and \eqref{estim RHS eq tffi} prove that there exists a unique $\left(\mathbf{\Phi}_1(\tffi) ,  \mathbf{\Phi}_2(\tW) \right)\in H^2_\delta\times H^2_\delta$ solving \eqref{eq Phi tW}-\eqref{eq Phi tffi} and satisfying
\begin{align*}
\l \mathbf{\Phi}_1(\tffi)  \r_{H^2_\delta} + \l  \mathbf{\Phi}_2(\tW)\r_{H^2_\delta} \lesssim C(C_1)\e^2 + \e.
\end{align*}
Therefore, if we take $C_1$ larger than the numerical constant appearing in these estimates and $\e$ small compared to $C_1$, then $\left(\mathbf{\Phi}_1(\tffi) ,  \mathbf{\Phi}_2(\tW) \right)\in \mathcal{B}\times \mathcal{B}$. This shows that $\mathbf{\Phi}$ is well-defined. 

\saut
In order to show that $\mathbf{\Phi}$ is a contraction we consider the equations satisfied by the differences $\mathbf{\Phi}_1(\tffi_a) -  \mathbf{\Phi}_1(\tffi_b)$ and $\mathbf{\Phi}_2(\tW_a)  -  \mathbf{\Phi}_2(\tW_b)$, where $(\tffi_a,\tW_a)$ and $(\tffi_b,\tW_b)$ are two elements of $\mathcal{B}\times \mathcal{B}$. Together with non-linear inequalities of the form
\begin{align*}
\left| x^k - y^k \right| \lesssim  \sup_{0\leq p,q\leq k-1}\left\{ |x|^p, |y|^q \right\} \times  |x-y|,
\end{align*}
we can mimick the previous arguments leading to \eqref{estim RHS eq tW} and \eqref{estim RHS eq tffi} and prove that by taking $C_1$ larger and $\e$ smaller if necessary the map $\mathbf{\Phi}$ is a contraction. We omit the details.
\end{proof}

The Banach fixed point theorem then implies that there exists $(\tffi,\tW)\in\mathcal{B}\times \mathcal{B}$ solving \eqref{eq tW final} and \eqref{eq tffi final}. We can also prove that $\tffi$ and $\tW$ enjoy higher regularity, as we did for $Y$ in Section \ref{section solving for Y}. We obtain $\tffi,\tW\in H^{N-3}_\delta$ with
\begin{align}
\l \tffi \r_{H^{k+2}_\delta} + \l \tW \r_{H^{k+2}_\delta} \lesssim \frac{\e}{\lambda^k}, \label{estim tffi et tW}
\end{align}
for $k\in \llbracket 0, N-5 \rrbracket$. This concludes the construction of high-frequency solutions to \eqref{hamiltonian constraint}-\eqref{momentum constraint}.

\section{Proof of the main theorem}\label{section conclusion chap 2}

In this section we conclude the proof of Theorem \ref{theo main chap 2}. The solution of the constraint equations $(\bar{g}_\lambda,K_\lambda)$ we constructed through the conformal method is given by
\begin{align}
\bar{g}_\lambda & =\ffi^4\gamma,  \label{g lambda final}
\\ K_\lambda & = \ffi^{-2}(\sigma+L_\gamma W)+\frac{1}{3}\ffi^4\gamma\tau, \label{K lambda final}
\end{align}
where $\gamma$, $\tau$, $\sigma$, $W$ and $\ffi$ are the parameters and unknowns of the conformal method and are defined in the previous sections. Let us check that the two previous expressions match the expressions of Theorem \ref{theo main chap 2} and the estimates therein.

\subsection{The metric $\bar{g}_\lambda$ and proof of \eqref{estim h bar}}

We start with the induced metric. Thanks to \eqref{ansatz ffi} and \eqref{g lambda final} we first have
\begin{align*}
\bar{g}_\lambda & = \bar{g}_0 + \lambda\gamma^{(1)} + \GO{\lambda^2},
\end{align*}
which matches \eqref{g bar theo chap 2} using \eqref{def gamma1} and \eqref{def Fbar 1}. If we now look at the order $\lambda^2$ or higher in $\ffi^4\gamma$, we see that it is composed of oscillating terms and terms satisfying better estimates:
\begin{align}
\left( \ffi^4\gamma \right)^{(\geq 2)} & = 4\ffi^{(2)}\bar{g}_0 + 4\tffi \bar{g}_0 + \gamma^{(2)} +
\lambda \left(  4\ffi^{(3)}\bar{g}_0 + 4\left(\tffi + \ffi^{(2)}\right) \gamma^{(1)}   \right) + \GO{\lambda^2},\label{ffi 4 gamma geq 2}
\end{align}
where the $\GO{\lambda^2}$ is a polynomial in terms of $\ffi^{(2)}$, $\tffi$, $\ffi^{(3)}$, $\bar{g}_0$, $\gamma^{(1)}$ and $\gamma^{(2)}$. Using \eqref{def ffi 2}, \eqref{def gamma2} and \eqref{def F21}-\eqref{def F22} we see that
\begin{align*}
\ffi^{(2)}\bar{g}_0 + \gamma^{(2)} & = \sin\left( \frac{u_0}{\lambda}\right) \bar{F}^{(2,1)} + \cos\left( \frac{2u_0}{\lambda}\right) \bar{F}^{(2,2)} .
\end{align*}
Therefore, by setting
\begin{align*}
\bar{\h}_\lambda & = \left( \ffi^4\gamma \right)^{(\geq 2)}  - 4 \ffi^{(2)}\bar{g}_0 - 4\gamma^{(2)},
\end{align*}
we prove that $\bar{g}_\lambda$ is indeed given by the expression \eqref{g bar theo chap 2}. We now prove estimate \eqref{estim h bar}. Thanks to \eqref{ffi 4 gamma geq 2} we have
\begin{align}
\bar{\h}_\lambda & = 4\tffi \bar{g}_0  +  \lambda \left(  4\ffi^{(3)}\bar{g}_0 + 4\left(\tffi + \ffi^{(2)}\right) \gamma^{(1)}   \right) + \GO{\lambda^2}.\label{h bar expansion}
\end{align}
The regularity of each term in $\bar{\h}_\lambda$ (recall \eqref{estim tffi et tW}) and the decay of $\tffi$ and $\bar{g}_0$ at infinity imply easily that the amount of derivatives together with the weights in \eqref{estim h bar} are allowed. The only part of \eqref{estim h bar} that remains to be checked is the $\lambda$ behaviour. From this perspective, $\ffi^{(3)}$, $\ffi^{(2)}$ and $\gamma^{(1)}$ are the worse terms since they lose one $\lambda$ power for each derivative. As they are already multiplied by $\lambda$ in \eqref{h bar expansion}, this concludes the justification of \eqref{estim h bar}.

\subsection{The tensor $K_\lambda$ and proof of \eqref{estim K geq 2}}

For the tensor $K_\lambda$, we first prove that \eqref{K lambda final} matches the expression \eqref{K lambda theo chap 2}. Since $\ffi=1+\GO{\lambda^2}$, we have $\ffi^{-2}=1+\GO{\lambda^2}$ and  $\ffi^{4}=1+\GO{\lambda^2}$. Therefore from \eqref{K lambda final} we obtain
\begin{align*}
K_\lambda & = \sigma^{(0)} + (L_\gamma W)^{(0)} + \lambda \left( \sigma^{(1)} + (L_\gamma W)^{(1)} + \frac{1}{3}\bar{g}_0 \tau^{(1)} \right) + \GO{\lambda^2}.
\end{align*}
We now use the ansatz for $W$ (see \eqref{ansatz W}) and the expansion of Lemma \ref{lem conformal killing} to obtain $(L_\gamma W)^{(0)}=0$ and $(L_\gamma W)^{(1)}= \mathbf{K}^{[-1]}(W^{(2)})$. This gives
\begin{align*}
K_\lambda & = \sigma^{(0)}  + \lambda \left( \sigma^{(1)} + \mathbf{K}^{[-1]}(W^{(2)}) + \frac{1}{3}\bar{g}_0 \tau^{(1)} \right) + \GO{\lambda^2}
\\& = K^{(1)}_\lambda + \lambda K^{(1)}_\lambda + \GO{\lambda^2},
\end{align*}
where we used the definition of $\sigma^{(0)}$ and $\sigma^{(1)}$, see \eqref{sigma 0} and \eqref{sigma 1}. Therefore, the solution $K_\lambda$ matches the expression \eqref{K lambda theo chap 2}. The remainder $K_\lambda^{(\geq 2)}$ satisfies
\begin{align}
K_\lambda^{(\geq 2)} & =  \sigma^{(2)}(Y) + L_\gamma Y + \mathbf{K}^{[\geq 0]}(W^{(2)}) +  \mathbf{K}^{[\geq -1]}(W^{(3)}(\tW)) + L_\gamma\tW  \label{K geq 2 expansion}
\\&\quad   -2\left( \tffi + \ffi^{(2)} \right) \sigma^{(0)}   +   \frac{1}{3}\gamma^{(1)}\tau^{(1)} + \GO{\lambda^3}.\nonumber
\end{align}
The estimate \eqref{estim K geq 2} then follows from estimating directly all the oscillating terms in \eqref{K geq 2 expansion} and using \eqref{estim Y} and \eqref{estim tffi et tW} for $L_\gamma Y$, $L_\gamma \tW$ or $\tffi$. This concludes the proof of Theorem \ref{theo main chap 2}.

\bibliographystyle{alpha}

\end{document}